\documentclass[journal]{IEEEtran}
\usepackage{epsfig,amsmath,amssymb,epsf,amsthm,scalefnt,multirow,subfig}
\usepackage{xcolor}
\usepackage{float}
\usepackage{cite}
\usepackage{psfrag}
\usepackage{tikz}
\usetikzlibrary{tikzmark,calc,decorations.pathreplacing}
\usepackage{amsmath}

\newtheorem{theorem}{Theorem}
\newtheorem{lemma}{Lemma}
\newtheorem*{lemma*}{Lemma}
\newtheorem{remark}{Remark}

\newtheorem{corollary}{Corollary}

  \def\cC{{\mathcal{C}}}

 \def\cN{{\mathcal{N}}}  \def\cP{{\mathcal{P}}}
 \def\cR{{\mathcal{R}}}  
 \def\cV{{\mathcal{V}}}

\def\sinr{\mathop{\mathrm{SINR}}}    

\def\limit{\mathop{\mathrm{lim}}}

\def\hh{\mathop{\mathrm{h}}}
\def\EE{\mathop{\mathrm{E}}}
\def\Var{\mathop{\mathrm{Var}}}

 \def\bomega  
\def\bTheta{{\pmb{\Theta}}}

\def\b0{{\pmb{0}}} 

\allowdisplaybreaks[4]

\begin{document}

\title{Mobility-Assisted Covert Communication over Wireless Ad Hoc Networks}

 \author{\IEEEauthorblockN{Hyeon-Seong Im$^*$ and Si-Hyeon Lee$^\dagger$}\\
\IEEEauthorblockA{$^*$Department of Electrical Engineering, POSTECH, South Korea\\
$^\dagger$School of Electrical Engineering, KAIST, South Korea\\
Emails: imhyun1209@postech.ac.kr, sihyeon@kaist.ac.kr}\thanks{The material in this paper will be presented in part at IEEE ISIT 2020 \cite{Im:2020}.}}

\maketitle

\begin{abstract}
We study the effect of node mobility on the throughput scaling of the covert communication over a wireless adhoc network. It is assumed that $n$ mobile nodes want to communicate each other in a unit disk while keeping the presence of the communication secret from each of $\Theta(n^s)$ non-colluding wardens ($s>0$). Our results show that the node mobility greatly improves the throughput scaling, compared to the case of fixed node location. In particular, for $0<s<1$, the aggregate throughput scaling is shown to be linear in $n$ when the number of channel uses that each warden uses to judge the presence of communication is not too large compared to $n$. 

For the achievability, we modify the two-hop based scheme by Grossglauser and Tse (2002), which was proposed for a wireless ad hoc network without a covertness constraint, by introducing a preservation region around each warden in which the senders are not allowed to transmit and by carefully analyzing the effect of covertness constraint on the transmit power and the resultant transmission rates. This scheme is shown to be optimal {for $0<s<1$} under an assumption that each node outside preservation regions around wardens uses the same transmit power.


\end{abstract}
\section{Introduction}\label{sec1}
In some communication applications like military communications, it is required not only that the adversary should not decode the message, 
but also that it should not detect the presence of communication. Such scenarios are called covert communications. 
The fundamental limits of covert communications have been characterized mainly for point-to-point scenarios such as additive white Gaussian noise (AWGN) channel \cite{Bash:2013, Ligong:2016,Bloch:2016}, discrete memoryless channel \cite{Ligong:2016,Bloch:2016}, channels with uncertainly \cite{Che:2014,Liu:2018,Shahzad:2017,Lee:2017}, and channels with multiple antennas \cite{Abdelaziz:2017}. For the standard AWGN channel \cite{Bash:2013,Ligong:2016,Bloch:2016} where a warden utilizes $l$ channel uses to judge the presence of the communication, it was shown that the received power at the warden should be  $\Theta(1/\sqrt{l})$ to satisfy the covertness, which in turn restricts the transmit power. Accordingly, the number of information bits that can be communicated covertly over AWGN channels with $l$ channel uses scales with $\sqrt{l}$ (called square root law). 

Recently, the study on the fundamental limits of covert communications has been extended to various multi-user scenarios, including broadcast channels \cite{Vincent:2019}, multiple access channels \cite{Keerthi:2016}, interference channels \cite{Cho:2020}, and some multi-hop networks \cite{Wu:2017}, \cite{Sheikholeslami:2018}. As a more general setup, the throughput scaling of the covert communication over a wireless ad hoc network was studied in \cite{Cho:2019} where $n$ nodes of fixed locations want to communicate each other covertly against a set of non-colluding wardens. The authors \cite{Cho:2019}  proposed multi-hop (MH) \cite{Gupta:2000} and hierarchical cooperation (HC) \cite{Ozgur:2007}-based schemes, which  are modified by introducing a preservation region \cite{Jeon:2011} around each warden. By preventing the transmission of the nodes inside the preservation regions, the nodes outside the preservation regions can increase the transmission power while satisfying the covertness constraint. 

In this paper, we study the effect of node mobility on the throughput scaling of the covert communication over the wireless adhoc network. For the case without a covertness constraint, it is known that the mobility of nodes does not increase the capacity scaling when the network area is fixed, i.e., the capacity scaling is linear in $n$ for both cases with and without the mobility of nodes \cite{David:2002,Ozgur:2007}. Hence, it would be an interesting question whether the mobility can increase the throughput scaling in the presence of  the covertness constraint. To that end, we assume that $n$ mobile nodes want to communicate each other in a fixed area while keeping the presence of the communication secret from each of $\Theta(n^s)$ wardens ($s>0$). The locations of wardens can be fixed or vary over time. A practical scenario of our model would be the military situation where several soldiers are invading the enemy's area while the soldiers are keeping the presence of the communication secret form each enemy. 

Interestingly,  we show that the mobility of nodes greatly improves the throughput scaling of the covert communication over the wirelss adhoc network, compared to the case of fixed node location \cite{Cho:2019}. In particular, for $s<1$, while the linear throughput scaling is not possible unless the path loss exponent $\alpha$ is very close to 2 for the case of fixed node location, it is possible with node mobility when the number of channel uses that each warden uses to  judge  is not too large compared to $n$. For the achievability, we propose a two-hop based scheme where a source node transmits its message to the nearest node, and the node who took over the message forwards it to the intended destination node once the two nodes become close. In our scheme, we also set preservation regions around each warden so the transmissions are not allowed inside those regions. The area of preservation regions is chosen to make the fraction of nodes inside the regions negligible. Note that in our scheme, the communication from a source to its destination consists of two-hop small-range transmissions by exploiting the node mobility. In contrast, the proposed schemes in \cite{Cho:2019} for the case of fixed node location involves long-range transmission (HC-based scheme) or multi-hops (MH-based scheme). We note that the long-range transmission of HC scheme does not degrade the performance in the absence of the covertness constraint as the received power is sufficiently large, but it  does degrade under the covertness constraint as the network turns into power-limited.

We note that our scheme operates similarly as the two-hop based scheme \cite{David:2002}, which was proposed for a wireless ad hoc network without a covertness constraint, but there are some technicalities different from \cite{David:2002} due to the presence of the covertness constraint. First, the received power at each warden is precisely evaluated to determine the allowable transmit power at the senders. Next, as the transmit power is severely constrained due to the covertness constraint, the distance between a sender-receiver pair affects the order of the point-to-point communication rate between them. By taking this fact into account, we carefully analyze the communication rate based on the distribution of the distance between a sender-receiver pair. 

The remaining of this paper is structured as follows. We introduce the network model and formulate the problem in Section~\ref{sec2}, and present the main results of this paper in Section~\ref{sec3}. To prove the main results, we first derive sufficient and necessary conditions for the covertness constraint in Section~\ref{sec4}  and then prove the achievability and the converse parts in Sections~\ref{sec5} and \ref{sec6}, respectively. Finally, Section~\ref{sec7} concludes the paper with some further works.
 
\section{Problem Statement}\label{sec2}

\subsection{Network Model}\label{sec2A}
In a unit disk, $n$ nodes are uniformly and independently distributed in each discrete time unit $t$.\footnote{Here each time $t$ consists of  several channel uses so that the communication rate of $\log (1+\sinr)$ between two communication parties is assumed be achievable for each time unit where $\sinr$ denotes the signal to noise plus interference ratio.} The random process governing the location of each node is assumed to be strict-sense stationary (SSS) and ergodic.
Each  node is a source  and a destination  simultaneously and the $n$ source-destination pairs are randomly determined. In the same area, there are $n_w=\Theta(n^s)$ for $s>0$ {non-colluding wardens.} We consider both cases where the wardens have mobility or not. For the case of no mobility, the wardens are uniformly and independently distributed and their locations are fixed across the time. For the other case, the location of the wardens can change in a SSS and ergodic manner. Each of the $n$ sources wants to communicate with its destination while keeping the presence of the communication secret from each warden. The covertness constraint is described in detail in the next subsection. The network  is illustrated in Fig.\ref{fig1}.

The received signal at  node $j$ at time $t$ is given as
\begin{align}
  Y_j[t] = \sum_{k=1}^{n}H_{jk}[t]X_k[t]+N_j[t],\label{eq:3}
\end{align}
where $X_k[t]$ is the transmitted signal by node $k$, $N_j[t]\sim\cC\cN(0,N_0)$ is the circular symmetric Gaussian noise with zero mean and variance $N_0$, and $H_{jk}[t]$ is the channel gain from node $k$ to node $j$ given by
\begin{align}
 H_{jk}[t] = {\sqrt{G}\over{(d_{jk}[t])^{\alpha/2}}}\exp(j\theta_{jk}[t]). \label{eq:1}
\end{align}
Here, $d_{jk}[t]$ is the distance between nodes $k$ and $j$, $\alpha>2$ is the path loss exponent, $\theta_{jk}[t]$ is uniformly and independently distributed phase, and $G$ is given as 
\begin{align}
  G = \left( {\lambda\over{4\pi}}\right)^2G_sG_r\label{eq:2}
\end{align}
from Friis' formula, 
where $G_s$ and $G_r$ are the antenna gains at the sender and the receiver, respectively, and $\lambda$ is the carrier wavelength.  Each  node has the same average power constraint  of $P$. The channel state information is available only at the receivers and the delay toleration of data packets from source to destination is assumed to be sufficiently large. We suppose that all the sender-receiver pairs share a sufficiently long secret key. 

The received signal at warden  $w$ at time $t$ is
\begin{align}
Z_w[t] = \sum_{k=1}^{n}H'_{wk}[t]X_k[t]+N'_w[t],\label{eq:4}
\end{align}
where $N'_w[t]\sim\cC\cN(0,N_0)$ is the circular symmetric Gaussian noise with variance $N_0$ and $H'_{wk}$ is the channel gain from  node $k$ to warden  $w$ defined in a similar manner as \eqref{eq:1}.

\begin{figure}
\centering
\includegraphics[width=0.9\columnwidth]{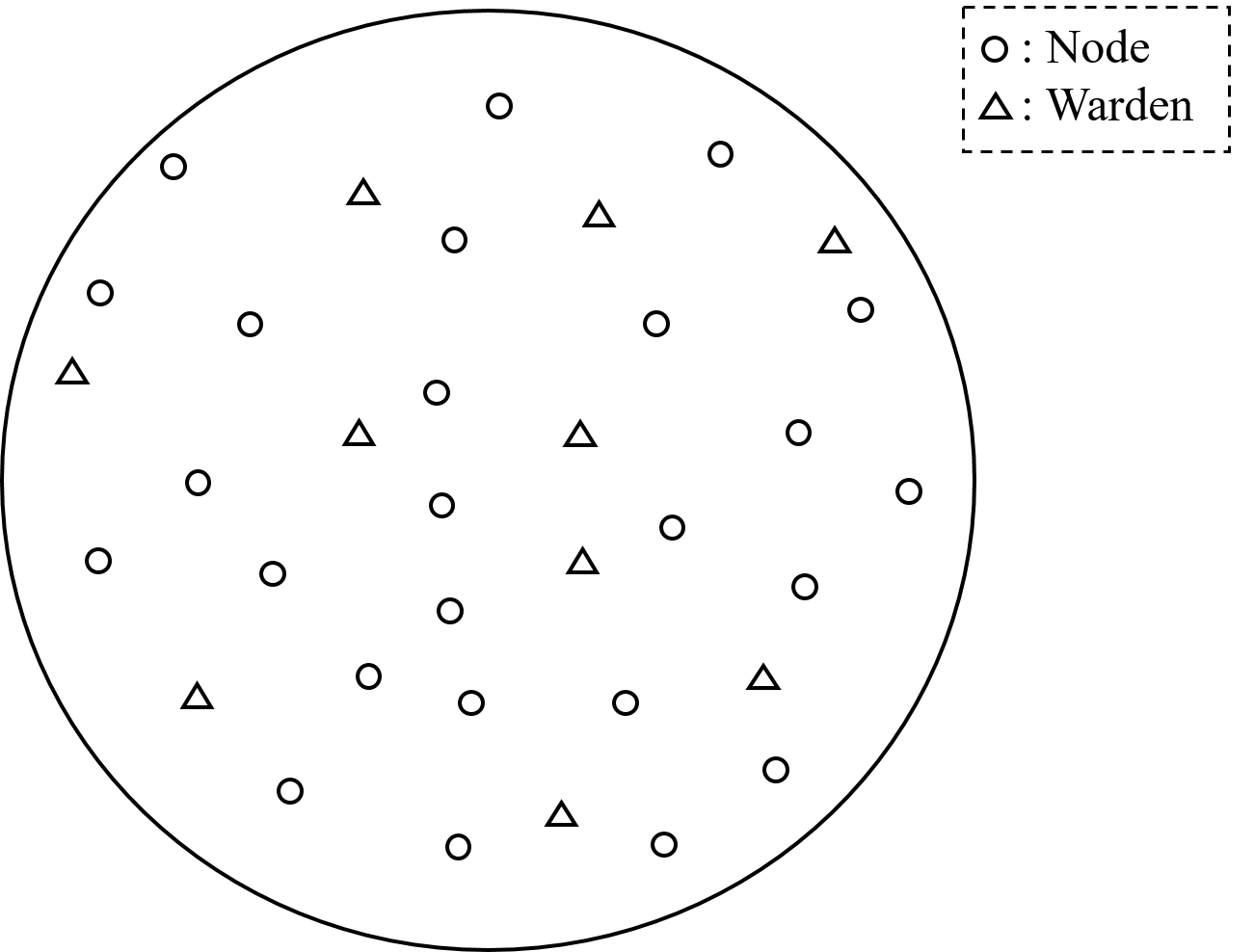}
\caption{In a unit area, $n$ mobile nodes want to communicate each other in a unit disk while keeping the presence of the communication secret from each of $\Theta(n^s)$ wardens. }\label{fig1}
\end{figure}
\subsection{Covertness Constraint}\label{sec2B} 
Each warden  tests the received signal over $l$ channel uses to detect whether the  nodes are communicating or not. The window of $l$ channel uses can be arbitrarily chosen over the whole communication.\footnote{The $l$ channel uses do not need to be consecutive in our analysis and the results continue to hold if the warden utilizes any arbitrary set of $l$ channel uses.} The communication is said to be covert if it is hard for the warden to determine whether the nodes are communicating (hypothesis $H_1$) or not (hypothesis $H_0$).  The optimal hypothesis test of warden $w$ satisfies
\begin{align}
 p(H_0|H_1)+p(H_1|H_0)&=1-V(Q_{Z_w^l} \| Q_{N'_{w}}^{\times l})\label{eq:5.2} \\ 
&\overset{(a)}\geq 1-\sqrt{D(Q_{Z_w^l} \| Q_{N'_{w}}^{\times l})}, \label{eq:5.3}
\end{align}
where $V(\cdot\|\cdot)$ is the total variational distance, $D(\cdot\|\cdot)$ is the Kullback-Leibler divergence, $Q_{Z_w^l}$ (resp. $Q_{N'_{w}}^{\times l}$) is the distribution of the received signal at warden $w$ over $l$ channel uses when the communication occurs (resp. does not occur). Here $N'_w\sim\cC\cN(0,N_0)$, and $Q_{N'_{w}}^{\times l}$ is the $l$-fold of $Q_{N'_w}$. Inequality $(a)$ is by the Pinsker's inequality, e.g.,~\cite{LehmannRomano:TSH}. Thus,  if $D(Q_{Z_w^l} \| Q_{N'_{w}}^{\times l})$ is small, the optimal hypothesis test of warden $w$ is similar with a blind test, which satisfies $p(H_0|H_1)+p(H_1|H_0)=1$. Hence, we set the covertness constraint as
\begin{align}
D(Q_{Z_w^l} \| Q_{N'_{w}}^{\times l})\leq\delta   \mbox{ for }   w=1,2,...,n_w .\label{eq:5}
\end{align}
for some $\delta>0$.

\subsection{Long-Term Throughput}\label{sec2C}
We note that the throughput from a source to its destination varies over time  since the  nodes have mobility. Let's say node $j$ communicates to its destination node $k$ at time $t$ with the rate of $R_{jk}(n,s,t)$, while satisfying the covertness constraint \eqref{eq:5} for all the wardens. The long-term throughput $\lambda(n,s)$ is said to be feasible if 
\begin{align}
\limit_{T\rightarrow\infty} {1\over T}{\sum_{t=1}^{T}R_{jk}(n,s,t)} \geq \lambda(n,s)\label{eq:6}
\end{align}
for all source-destination pairs $(j,k)$. The goal of this paper is to characterize the scaling of the maximally achievable aggregate throughput. 

\section{Main results}\label{sec3}
In this section, we state our main results. First, achievable aggregate throughputs are presented in Theorems \ref{thm1} and \ref{thm2}, which are proved in Section \ref{sec5}. For the converse, a trivial upper bound on the aggregate throughput is given in Theorem \ref{thm_ub} and some nontrivial upper bounds under an assumption are given in Theorems \ref{thm3} and \ref{thm4}, which are proved in Section \ref{sec6}. 
\begin{theorem}\label{thm1}
For $0< s<1$, the following aggregate throughput is achievable with covertness constraint $\delta$ and testing channel length $l$ for any $\epsilon>0$:
\begin{align}
T(n,s)=\Theta\left(n^{1-\epsilon}\cdot \min \left(\left({n^{(1/2-s/2)(\alpha-2)}}\over \sqrt{l}\right)^{2/\alpha}, 1 \right) \right),\label{eq:7}
\end{align}
with high probability (probability going to $1$ as $n$ goes to infinity).
\end{theorem}

\begin{theorem}\label{thm2}
For $s\geq1$, the following aggregate throughput is achievable with covertness constraint $\delta$ and testing channel length $l$ for any $\epsilon>0$:
\begin{align}
T(n,s) = \Theta\left(n^{1-\epsilon}\cdot  \left({n^{{{\alpha}(1/2-s/2)}}\over\sqrt{l}}\right)^{2/\alpha}\right),\label{eq:9}
\end{align}
with high probability.
\end{theorem}

It is known that the aggregate throughput of $\Theta(n^{1-\epsilon})$ for any $\epsilon>0$ is achievable for a wireless adhoc network without any covertness constraint \cite{David:2002}. In the presence of the covertness constraint, our achievable aggregate throughput becomes $\Theta(n^{1-\epsilon})$ when the number of nodes is smaller than the number of wardens and the testing channel length $l$ is sufficiently small compared to the number of nodes. This makes sense since the transmission power would be restricted more severely as more wardens observe more channel outputs. 

Let us briefly describe our proposed scheme and provide a sketch of the proof. The details are in Section \ref{sec5}. Our scheme operates in two phases similarly as the scheme in \cite{David:2002}. In each time slot $t$, a certain fraction of $n$ nodes operate as senders and the others as potential receivers. Each sender communicates with the nearest receiver (sender-receiver pair). The senders and the receivers play the roles of sources and relays, respectively, in phase 1 (odd times), and the roles of relays and destinations, respectively, in phase 2 (even times).  In phase 1, each node selected as a sender transmits its own source data packet to the receiver. In phase 2, each node selected as a sender selects and transmits the data packet intended for the receiver among the stored data packets. In this process, a sender does not transmit if it is inside a certain area around any warden, which we call a preservation region. The area of preservation regions is chosen to make the fraction of nodes inside preservation regions negligible.

Now, the throughput scalings in Theorems \ref{thm1} and \ref{thm2} are roughly derived as follows (a rigorous proof is in Section \ref{sec5}):
\begin{align}
T(n,s)& \approx{1\over2}\cdot{\EE} \left[\sum_k  \log (1+\sinr(r_k)) \right]\label{eq:9.01}\\
&= {1\over2}\cdot {\EE} \left[\sum_{k: \sinr(r_k)=\Omega(1)} \log (1+\sinr(r_k)) \right. \cr
&~~~+ \left. \sum_{k: \sinr(r_k)=o(1)} \log (1+\sinr(r_k))\right] \\
&\overset{(a)}\gtrsim{1\over2}\cdot{\EE} \left[ \sum_{k: \sinr(r_k)=\Omega(1)} \log (1+\sinr(r_k)) \right] \\
&\overset{(b)}\gtrsim n\cdot p(\sinr(r)=\Omega(1)) \Theta(1)\\ 
&\overset{(c)}\approx{n \cdot p(r \leq r_m) \cdot \Theta(1)}\label{eq:9.1}\\
&\overset{(d)}\approx {n \cdot \min(n r_m^2,1) \cdot \Theta(1),}\label{eq:9.2}
 \end{align}
where the expectations are with respect to the random locations of nodes, $r_k$ is the distance between the $k$-th sender-receiver pair, $r$ is the distance between a randomly chosen sender-receiver pair, and $r_m$ is the distance between the sender-receiver pair when $\sinr(r_m)=\Theta(1)$. Here, we have the factor of  $\frac{1}{2}$ since our scheme operates in two phases. We can show that $(a)$ is tight in the sense of scaling by noting that the probability that a sender-receiver pair has a distance of order $r$ is proportional to $r^2$ and $\sinr$ is proportional to $r^{-\alpha}$, and $(b)$ is tight up to the logarithmic order. Also, $(c)$ follows from the property that $\sinr(r)$ is decreasing function in $r$ and $(d)$ is because the probability that a sender-receiver pair has a smaller distance than $r_m$ is proportional to the product of $r_m^2$ and $n$ since the corresponding receiver to a sender is the nearest receiver to  the sender among  other receivers, and because the probability does not exceed 1.  The throughput scailngs in  Theorems \ref{thm1} and \ref{thm2} are obtained by deriving $r_m$ for each case of $0<s<1$ and $s\geq 1$, while Theorem \ref{thm2} has no minimum term because $n r_m^2$ cannot exceed $1$ in that case.

Next, the following is a trivial upper bound on the aggregate throughput scaling, which is the upper bound without the covertness constraint. 
\begin{theorem}\label{thm_ub}
For $s\geq 0$, the aggregate throughput  with covertness constraint $\delta$ and testing channel length $l$ is upper-bounded as follows  for any $\epsilon>0$: 
\begin{align}
T(n,s)=O\left(n^{1+\epsilon} \right),\label{eq:ub}
\end{align}
with high probability.
\end{theorem}
Hence, for $0<s<1$, if the testing channel length $l$ is sufficiently small so that the minimum in \eqref{eq:7} becomes one, the aggregate throughput scaling in Theorem \ref{thm1} is tight. 

The difficulty in proving a non-trivial upper bound comes from the fact that the distances between the senders and the wardens, which are related to the upper bound on the transmit power from the covertness constraint, and the distances between the senders and the receivers, which affect the transmission rate, independently vary over time. The optimal transmit power control in such a scenario seems to be a challenging problem. As an alternative, for the upper bound, we assume that each node distant from every warden to a certain extent uses the same power at each channel use. 

\begin{theorem}\label{thm3}
For $0<s<1$, the aggregate throughput  with covertness constraint $\delta$ and testing channel length $l$ is upper-bounded as follows for any $\epsilon>0$ under the assumption that each node not contained in the regions of radius $\Theta(n^{-(\frac{s}{2}+\epsilon')})$ around each warden for an arbitrarily small $\epsilon'>0$ uses the same power at each channel use: 
\begin{align}
T(n,s)=O\left(n^{1+\epsilon}\cdot \min \left(\left({n^{(1/2-s/2)(\alpha-2)}}\over \sqrt{l}\right)^{2/\alpha}, 1 \right) \right),\label{eq:10}
\end{align}
with high probability.
\end{theorem}
\begin{theorem}\label{thm4}
For $s\geq1$, the aggregate throughput with covertness constraint $\delta$ and testing channel length $l$ is upper-bounded as follows for any $\epsilon>0$ under the assumption that each node not contained in the regions of radius $\Theta(n^{-(\frac{s}{2}+\epsilon')})$ around each warden for an arbitrarily small $\epsilon'>0$ uses the same power at each channel use: 
\begin{align}
T(n,s)=O\left(n^{1+\epsilon}\cdot \left({1\over\sqrt{l}}\right)^{2/\alpha}\right),\label{eq:12}  
\end{align}
with high probability.
\end{theorem}
We note that the aggregate throughput scaling in Theorem \ref{thm1} for $0<s<1$ is tight under the assumption that each node not contained in the regions of radius $\Theta(n^{-(\frac{s}{2}+\epsilon')})$ around each warden  for an arbitrarily small $\epsilon'>0$  uses the same power at each channel use. The radius of $\Theta(n^{-(\frac{s}{2}+\epsilon')})$  is set as the same with that of the preservation region introduced in our achievability scheme. Hence, for $0<s<1$, our proposed scheme is scaling-optimal if the nodes outside the preservation regions are not allowed to change the transmit power over time. However, there exists a gap between the lower and upper bounds in Theorems \ref{thm2} and \ref{thm4} for $s\geq 1$. Such a non-tightness comes from the gap between the pessimistic and optimistic derivations of the distance between a sender and the nearest warden.

\begin{remark}\label{remark_1}
Let us compare our results to the case without node mobility \cite{Cho:2019}, which proposed HC \cite{Ozgur:2007} and MH \cite{Gupta:2000} based schemes. 
Interestingly, our throughput scaling with node mobility is strictly higher than that without mobility. The HC scheme contains long-range MIMO transmissions and hence the received SNR is much smaller compared to short-range transmissions for the same transmit power determined from the covertness constraint, which degrades the throughput scaling and results in a throughput gap that increases as $\alpha$ increases. On the other hand, the MH scheme consists of small-range transmissions, but it requires multiple hops for a source message to finally reach its destination, which degrades the throughput scaling and results in the throughput gap of $\Omega(n^{1/2})$. Our scheme performs better as the communication from a source to its destination consists of two-hop small-range transmissions, by exploiting the node mobility. 
\end{remark}

\section{Sufficient and Necessary Conditions of Covertness}\label{sec4}
In this section, we derive sufficient and necessary conditions for the covertness constraint, which are used for the derivations of lower and upper bounds on the aggregate throughput scaling, respectively. 
\subsection{Sufficient Condition of Covertness}\label{sec4A}
In the proposed scheme in Section \ref{sec5B}, each node uses i.i.d. complex Gaussian codebook. For such a code,  the Kullback-Leibler divergence \eqref{eq:5} is upper-bounded as:
\begin{align}
D(Q_{Z_w^l} \| Q_{N'_{w}}^{\times l}) &\overset{(a)}=\sum_{u=1}^{l}D(Q_{Z_{w,u}} \| Q_{N'_{w}}) \\
&= \sum_{u=1}^{l}D(\cC\cN(0,\rho_{w,u}+N_0) \| \cC\cN(0,N_0))\label{eq:13.1} \\
&\leq l\cdot D(\cC\cN(0,\rho_{wm}+N_0) \| \cC\cN(0,N_0)) \label{eq:13}\\
&\overset{(b)}= l \cdot \left({\rho_{wm}\over N_0}-\log{{N_0+\rho_{wm}}\over N_0}\right)\label{eq:14}   \\
&\overset{(c)}\leq l \cdot \left({\rho_{wm}\over N_0}-\left({\rho_{wm}\over N_0}-{\rho_{wm}^2\over{2 N_0^2}}\right)\right)\label{eq:15} \\
&= {{l\cdot \rho_{wm}^2}\over{2 N_0^2}},\label{eq:16} 
\end{align}
where $Z_{w,u}$ is the received signal of warden $w$ at channel use $u$, $\rho_{w,u}$ is the power of $Z_{w,u}$, and $\rho_{wm}$ is defined by $\max(\rho_{w,1}, ...,\rho_{w,l} )$. Here, $(a)$ is due to the use of the  i.i.d. complex Gaussian codebook, $(b)$ is proved in {\cite[Equation (80)]{Ligong:2016}}, and  $(c)$ follows from $\log(1+a) \geq a -{a^2\over 2}$ for $a>0$.

From \eqref{eq:16}, the covertness constraint \eqref{eq:5} is satisfied if the following inequality holds
\begin{align}
\rho_{wm} &\leq \sqrt{2} N_0 \sqrt{\delta \over l} .\label{eq:18}
\end{align}
\subsection{Necessary Condition of Covertness}\label{sec4B}
 The Kullback-Leibler divergence \eqref{eq:5} can be lower-bounded as:
\begin{align}
D&(Q_{Z_w^l} \|Q_{N'_{w}}^{\times l}) = -\hh (Z_w^l)+{\EE}_{{Z_w^l}}\left[\log{1 \over Q_{N'_w}^{\times l}(Z_w^l)}\right]\label{eq:19} \\
&= \sum_{u=1}^{l}\left( -\hh (Z_{w,t}|Z_w^{u-1})+{\EE}_{{Z_{w,u}}}\left[\log{1 \over Q_{N'_{w}}(Z_{w,u})}\right]\right)\label{eq:20}\\
&\geq \sum_{u=1}^{l}\left( -\hh (Z_{w,u})+{\EE}_{{Z_{w,u}}}\left[\log{1 \over Q_{N'_{w}}(Z_{w,u})}\right] \right)\label{eq:21} \\
&=\sum_{u=1}^{l} D(Q_{Z_{w,u}} \|Q_{N'_{w}}) \label{eq:22}\\
&\overset{(a)}\geq l \cdot D(Q_{\bar Z_{w}} \|Q_{N'_{w}}),\label{eq:23}
\end{align}
where $Z_w^{u-1}$ is the received signal of warden $w$ up to channel use $u-1$ and $Q_{\bar Z_{w}} = {1\over l}\sum_{u=1}^{l} Q_{Z_{w,u}}$ is the average distribution of received signal at warden $w$ over $l$ channel uses. The inequality $(a)$ is due to the convexity of Kullback-Leibler divergence.   

By \eqref{eq:23}, if the covertness constraint \eqref{eq:5} is satisfied, then it implies
\begin{align}
D(Q_{\bar Z_{w}} \|Q_{N'_{w}}) \leq {\delta \over l}.\label{eq:24}
\end{align}
Furthermore, the marginalized covertness constraint \eqref{eq:24} can be lower-bounded as:
\begin{align}
D(Q_{\bar Z_{w}} &\|Q_{N'_{w}}) =  -\hh({\bar Z_w})+{\EE}_{\bar Z_w}\left[\log{1 \over Q_{N'_{w}}(\bar Z_{w})}\right]\label{eq:25} \\
&= -\hh({\bar Z_w}) + {\EE}_{\bar Z_w}\left[\log{\left(\pi N_0 \exp\left({|\bar Z_w|^2 \over N_0}\right)\right)}\right] \label{eq:26}\\
&= -\hh({\bar Z_w}) + \log{\pi N_0} + {\EE}_{\bar Z_w}\left[{|\bar Z_w|^2 \over N_0}\right] \label{eq:27}\\
&\overset{(a)}= -\hh({\bar Z_w}) + \log{\pi N_0} + {{\bar \rho_w + N_0 }\over N_0} \label{eq:28}\\
\begin{split}\label{eq:29}
&\overset{(b)}\geq -\log{(\pi e(\bar\rho_w+N_0))}+\log{\pi N_0}\\
&+{{\bar \rho_w + N_0 }\over N_0}
\end{split}
\\
&={{\bar \rho_w }\over N_0} -\log{{\bar\rho_w+N_0}\over N_0},\label{eq:30}
\end{align}
where $\bar\rho_w =  {1\over l}\sum_{u=1}^{l} \rho_{w,u} $ is the average received power at warden $w$ over $l$ channel uses. Here, $(a)$ is because ${\EE}_{\bar Z_w}[{|\bar Z_w|^2 }]$ is same with ${\bar \rho_w + N_0 }$ since $\bar Z_w$ has the average distribution of $Z_{w,1},...,Z_{w,l}$ and $(b)$ is since the differential entropy is maximized with Gaussian distribution when the second moment is fixed. Since ${\bar \rho_w }$ goes to zero as $l$ goes to infinity, by Taylor expansion,  \eqref{eq:30} becomes
\begin{align}
D(Q_{\bar Z_{w}} \|Q_{N'_{w}}) \geq {{\bar \rho_w^2 }\over 2N_0^2} + o({\bar \rho_w^2 }).\label{eq:31}
\end{align}
Consequently, by \eqref{eq:24} and \eqref{eq:31}, if the covertness constraint \eqref{eq:5} is satisfied, then it implies
\begin{align}
{\bar \rho_w } \leq \sqrt{2}N_0\sqrt{\delta \over l}+o(l^{-1/2}).\label{eq:32}
\end{align}
\section{Achievability}\label{sec5}
In this section, we prove Theorems \ref{thm1} and \ref{thm2}. We first introduce preservation regions in Section \ref{sec5A} and explain our  two-hop scheme in Section \ref{sec5B}. Then,  the achievable aggregate throughput scalings are  derived in Sections \ref{sec5C} and \ref{sec5D}. We note that our proof does not rely on whether the wardens have mobility or not.

\subsection{Preservation Region}\label{sec5A}
In Section \ref{sec4A}, we show that a  sufficient condition for the covertness constraint is to make the received power at each warden less than some threshold. To increase the transmission power while keeping the received power at the wardens small, we introduce a preservation region of certain radius around each warden, in which the senders do not transmit. Let $r_p$ denote the radius of the preservation regions. As we increase $r_p$, the nodes outside the preservation regions can transmit with a higher power, but the area $\epsilon(n_w, r_p)=\Theta(n_wr_p^2)$ of all the preservation regions also increases. Hence, we set $r_p= \Theta(n^{-({s\over2}+\epsilon)})$ for an arbitrarily small $\epsilon>0$, which is the maximum radius while satisfying $\epsilon(n_w, r_p)\rightarrow 0$. 
In the scenarios \cite{Cho:2019}, \cite{Jeon:2011} where the node locations are fixed, such an introduction of preservation regions results in outage, i.e., $\epsilon(n_w, r_p)$ fraction of nodes do not participate in the whole communication at all. However, for our scenario where the nodes have mobility, each node is  in outage only for $\epsilon(n_w, r_p)$ fraction of time and hence there are no nodes in outage. 
\subsection{Two-Hop Scheme}\label{sec5B}
In the following, we explain our two-hop scheme. The scheme is modified from the scheme in \cite{David:2002} by considering the covertness constraint and preservation regions. 
\begin{itemize}
    \item For each time $t$, divide $n$ nodes into $\theta n$ senders and $(1-\theta)n$ receivers, where $0<\theta<1$ is fixed for the whole communication time.
    \item The sender-receiver pairs are determined in a way that each sender is paired with the nearest receiver for each time $t$. This way is called as sender-centric and described in Fig. \ref{fig2}.
    \item Each sender not in a preservation region transmits to its receiver using an i.i.d. complex Gaussian codebook with zero mean and variance of $P_{\rm{tx}}$, which is determined in Section \ref{sec5C} to satisfy the covertness constraint. The senders in a preservation region do not transmit. 
    \item The overall communication is divided into two phases: phase 1 is activated in odd times and phase 2 is activated in even times.  The data packets that the senders transmit to their receivers in each phase are as follows.  
    \begin{enumerate}
        \item Phase 1: The senders and the receivers play the roles of sources and relays, respectively. As shown in Fig. \ref{fig3}, each sender transmits its own source data packets (that have not been transmitted before) to its receiver. If the receiver is the final destination, this corresponds to direct transmission. Otherwise, the receiver keeps this packet and relays to the final destination in phase 2. 
        \item Phase 2: The senders and the receivers play the roles of relays and destinations, respectively. As shown in Fig. \ref{fig4}, each sender selects and transmits the data packet destined for the receiver among the stored data packets. Direct transmission is also possible if the sender is the source node of the receiver node. 
    \end{enumerate}
\end{itemize}


\begin{figure} 
\centering
 \includegraphics[width=0.9\columnwidth]{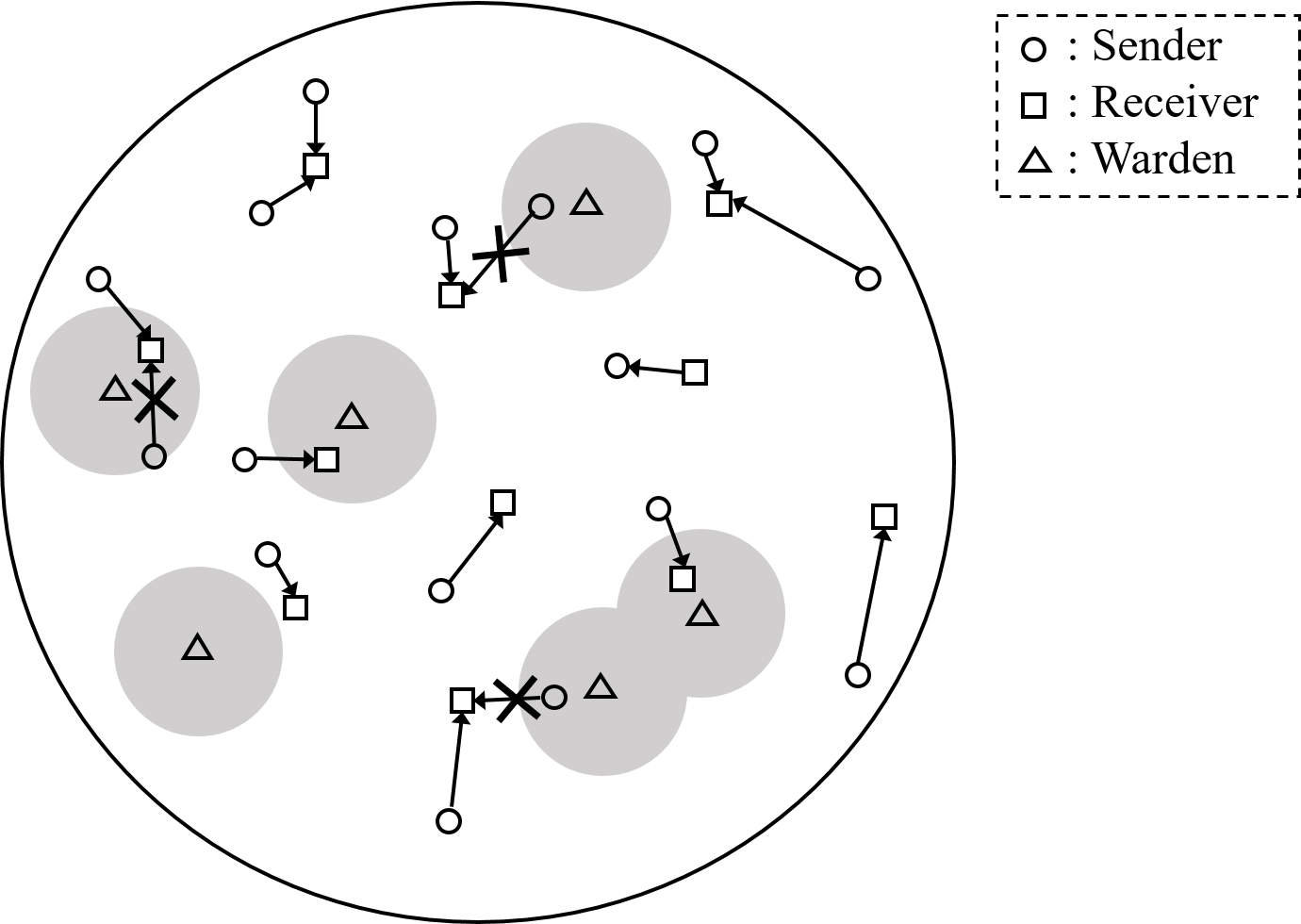}
\caption{Each sender transmits data packets to the nearest receiver when the sender is not in a (shaded) preservation region.}\label{fig2}
\end{figure}

We note that this two-hop scheme exploits multi-user diversity by letting some nodes operate as relays. In the absence of the covertness constraint,  it was shown in  \cite{David:2002} that such a scheme can achieve the maximum throughput $\Theta(n)$ because the source can transmit data packets independently of the distance to the destination by transmitting to a relay close to the source.  In \cite{David:2002}, it is shown that the aggregate throughput of  the direct transmission scheme, where each source transmits data packets only when the destination is sufficiently close, cannot achieve $\Theta(n)$ because it dose not use the multi-user diversity in transmitting data packets.

Although the network operation of our scheme seems to be similar with that in  \cite{David:2002}, there are some technicalities different from  \cite{David:2002} due to the presence of the covertness constraint. First, the received power at each warden is precisely evaluated to determine the allowable transmit power at the senders in Section~\ref{sec5C}. As the transmit power is severely constrained due to the covertness constraint, the distance between a sender-receiver pair affects the order of the point-to-point communication rate between them. By taking this fact into account, we carefully analyze the communication rate based on the distribution of the distance between a sender-receiver pair in Section \ref{sec5D}. 

Note that at the beginning of the communication, the senders might have no data packet destined for the receiver in phase~2. However, as the communication process proceeds and phase~1 takes place  several times, every node gradually collects and stores data packets of all the other nodes. In the following proof, we assume such a steady state and hence assumes that each sender has a data packet destined for the receiver in phase~2.

\begin{figure}
\centering
\includegraphics[width=0.9\columnwidth]{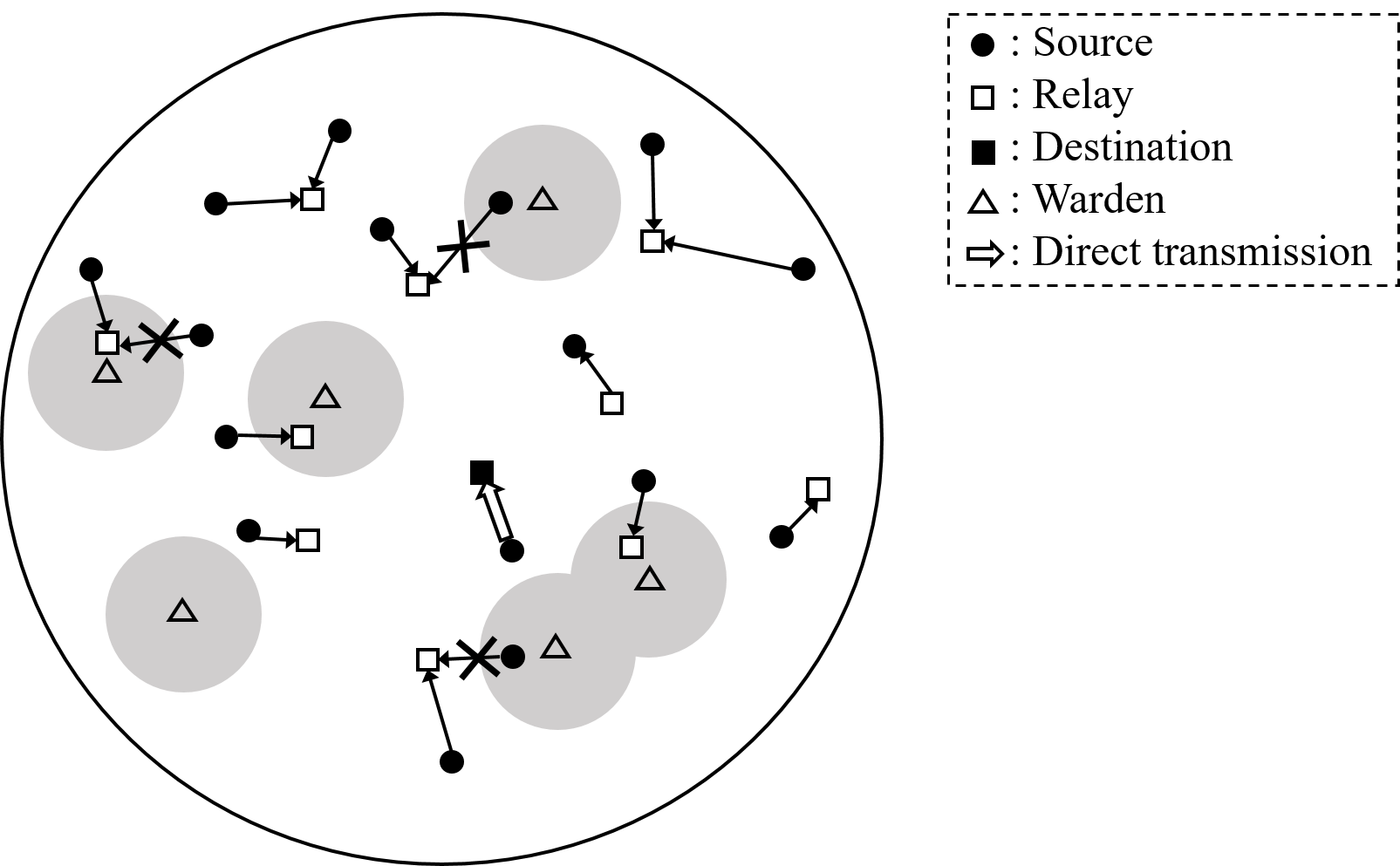}
\caption{In phase 1, the senders and the receivers play the roles of sources and relays, respectively. Each sender, not in a (shaded) preservation region, transmits its own source data packets to its receiver. If the receiver is the final destination, this corresponds to direct transmission. }\label{fig3}
\end{figure}
\begin{figure}
\centering
\includegraphics[width=0.9\columnwidth]{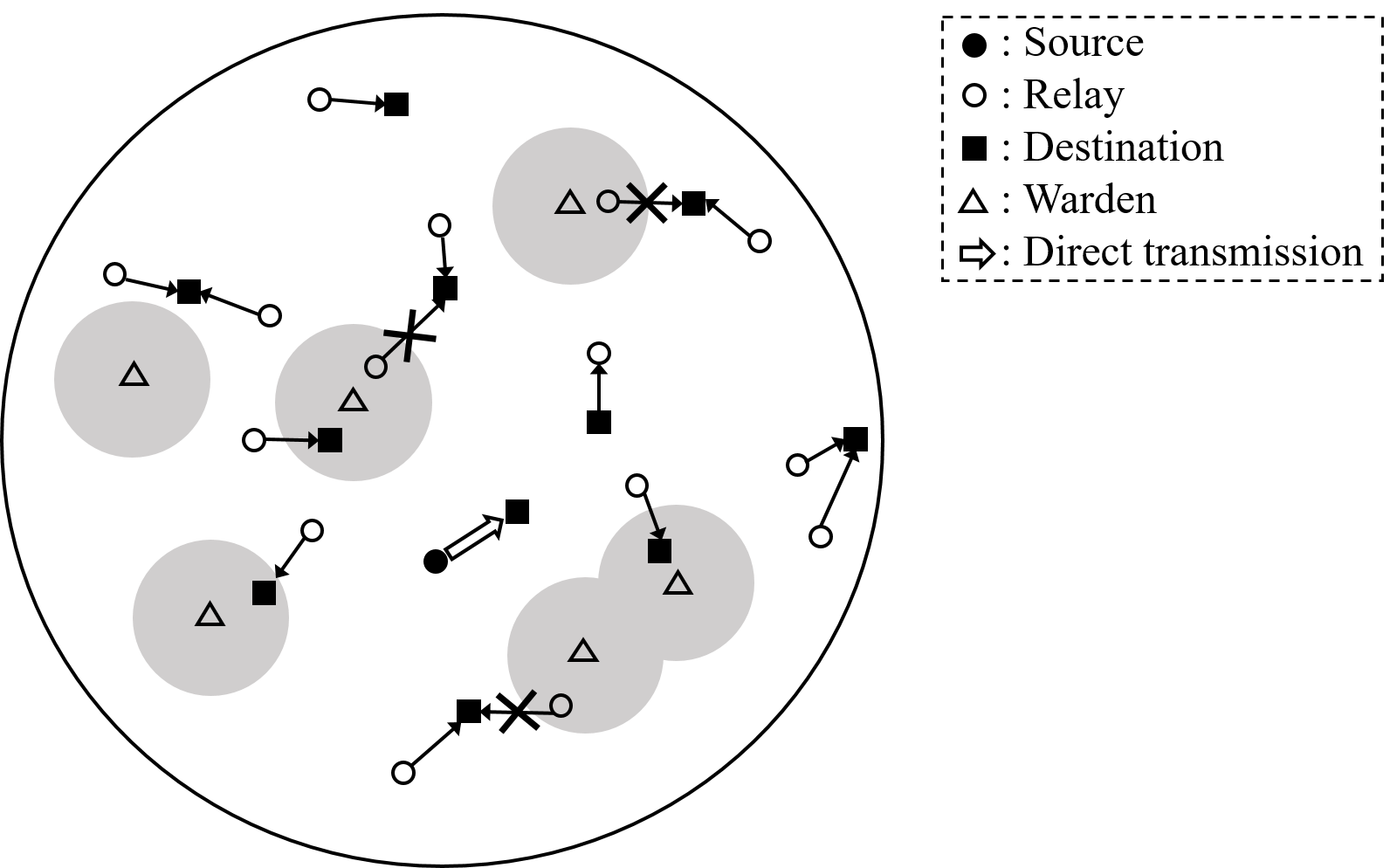}
\caption{In phase 2, the senders and the receivers play the roles of relays and destinations, respectively. Each sender, not in a (shaded) preservation region, selects and transmits the data packet destined for the receiver among the stored data packets. Direct transmission is also possible if the sender is the source node of the receiver node. }\label{fig4}
\end{figure}

\begin{remark}\label{remark_3}
The sender-receiver pairs can be alternatively determined in a way that each receiver is paired with the nearest sender for each time $t$ (receiver-centric). It can be checked that Theorems \ref{thm1} and \ref{thm2} continue to hold for the scheme operating according to the receiver centric basis.
\end{remark}

\subsection{Allowable Transmission Power}\label{sec5C}
In  this subsection, we derive an allowable transmission power. We start by stating how many nodes are inside a certain area with high probability. 
\begin{lemma}\label{lem1}
Suppose that $n$ nodes are uniformly and independently distributed in a unit area. Then, the number of nodes in a region of area $A(n)$ is between $((1-\delta)nA(n),(1+\delta)nA(n))$ with a probability larger than $1-{{1-A(n)}\over{\delta^2 n A(n)}}$ for any $\delta>0$.
\end{lemma}
\begin{proof}
We note that the number of nodes in a region of area $A(n)$ corresponds to a sum of  i.i.d. Bernoulli random variables \cite{Ozgur:2007}. The probability that a node is in the region is $A(n)$. Hence, the number of nodes in the region can be expressed as $\sum_{i=1}^{n}B_i$ where $B_i$'s are i.i.d. Bernoulli random variables with $p(B_i=1)=A(n)$. Then,
\begin{align}
&p(|\mbox{number of nodes in }A(n)-nA(n)|>\delta nA(n)) \label{eq:36}\\
&=p\left(\left|\sum_{i=1}^{n}B_i-nA(n)\right|>\delta nA(n)\right)\label{eq:38} \\
\begin{split}\label{eq:39}
&=p\left(\left|\sum_{i=1}^{n}B_i-\EE\left(\sum_{i=1}^{n}B_i\right)\right|    \right. \\
&\ \ \ \ \ \ \ \ \ \ \ \ \  \left. >\frac{{\delta\sqrt{ nA(n)}}}{\sqrt{1-A(n)}} \sqrt{\Var \left(\sum_{i=1}^{n}B_i\right)}\right)
\end{split}\\
&\overset{(a)}<{{1-A(n)}\over{\delta^2 nA(n)}},\label{eq:40}
\end{align}
for any $\delta>0$, where $\EE(\sum_{i=1}^{n}B_i)=nA(n)$, and $\Var (\sum_{i=1}^{n}B_i)=nA(n)(1-A(n))$ because $\Var(B_i)=A(n)(1-A(n))$ and $B_1, ..., B_n$ are  i.i.d. random variables. The inequality $(a)$ is by Chebyshev's inequality. 
\end{proof}

The following corollary is a direct consequence of Lemma~\ref{lem1}.
\begin{corollary}\label{cor1}
Suppose that  $n$ nodes are uniformly and independently distributed in a unit area. Then, the number of nodes in the region of area  $A(n)=\omega(1/n)$ is between $((1-\delta)nA(n),(1+\delta)nA(n))$ with high probability for any $\delta>0$.
\end{corollary}

Now, we show an allowable transmission power from the sufficient condition for the covertness constraint in Section \ref{sec4A} by deriving an upper bound on the received power at a warden. To derive the upper bound,  we first consider a set of disjoint rings, centered at the warden, that covers the whole network as shown in Figs. \ref{fig5} and \ref{fig6}.  Then we add up the received power at the warden that the senders in each ring contribute to. In this procedure, we bound the number of senders in each ring by using Corollary \ref{cor1}. Since the corollary is applicable only when the area is $\omega(1/n)$,  the width of the smallest ring depends on whether the area of a preservation region is $\omega(1/n)$ or not (corresponding to whether  $0<s<1$ or $s\geq 1$), which results in different bounds. The following lemmas present allowable transmission powers satisfying the covertness constraint for $0<s<1$ and $s\geq 1$. 

\begin{lemma}\label{lem2}
Let each node transmit using complex Gaussian coodbook with zero mean and variance of $P_{\rm{tx}}$. If $0<s<1$, then $P_{\rm{tx}}=\Theta(l^{-{1\over2}}{n^{-({s\over2}(\alpha-2)+1)-\epsilon}})$ {for any $\epsilon>0$} satisfies the covertness constraint with high probability.
\end{lemma}
\begin{proof}
Let $s_i(u)$ and $c_w(u)$ be the locations of sender $i$ and warden $w$ at the $u$-th channel use the warden observes, respectively. As shown in Fig. \ref{fig5}, we consider a set of disjoint rings with width $n^{-1/2}$, centered at the warden, that covers the whole network except the preservation region around warden $w$. Let $R_{1\beta}$ denote the $(\beta+1)$-th smallest ring. Then, the received power at warden $w$ at channel use $u$, $\rho_{w,u}$, is upper-bounded as:
\begin{align}
&\rho_{w,u}\overset{(a)}\leq \sum_{i:|s_i(u)-c_w(u)|\geq r_p}P_{\rm{tx}}(i,u)\cdot {G\over{|s_i(u)-c_w(u)|^\alpha}}\label{eq:41}\\
&\overset{(b)}\leq\sum_{\beta=0}^{\kappa_1(n,u)}\sum_{i\in \cR_{1\beta}}P_{\rm{tx}}(i,u)\cdot {G\over(r_p+\beta n^{-1/2})^\alpha}\label{eq:41.1}\\
&\overset{(c)}\leq\sum_{\beta=0}^{\kappa_1(n,u)}P_{\rm{tx}}(1+\epsilon'')\cdot {G\over(r_p+\beta n^{-1/2})^\alpha}\cdot |\cR_{1\beta}|\label{eq:42}\\
&\overset{(d)}\leq\sum_{\beta=0}^{\kappa_1(n,u)}P_{\rm{tx}}(1+\epsilon'')\cdot {G\over(r_p+\beta n^{-1/2})^\alpha}\cdot (1+\delta)\theta nA(R_{1\beta})\label{eq:42.1}\\
\begin{split}\label{eq:43}
&\overset{(e)}\leq\sum_{\beta=0}^{\kappa_1(n,u)}P_{\rm{tx}}(1+\epsilon'')\cdot {G\over(r_p+\beta n^{-1/2})^\alpha}\\
& \ \ \ \  \cdot (1+\delta)\theta n 2\pi(r_p+(\beta+1)n^{-1/2})n^{-1/2}
\end{split}\\
&\overset{(f)}\leq\int_{r_p-n^{-1/2}}^{{1/\sqrt{\pi}}+|c_w(u)|} \!\!\!\!\!\!\!\!\!\!\!\!\!\!\!\!\!\!\! P_{\rm{tx}}(1+\epsilon'')\cdot{G\over x^\alpha}\cdot (1+\delta)\theta n2\pi(x+n^{-1/2})dx\label{eq:44}\\
&\leq K_1 P_{\rm{tx}} n r_p^{2-\alpha},\label{eq:45}
\end{align}
with high probability {for any $\epsilon''>0$ and $\delta>0$}, where $P_{\rm{tx}}(i,u)$ is the transmission power of sender $i$ at channel use $u$, $\kappa_1(n,u)={\lfloor{{1/\sqrt{\pi}+|c_w(u)|-r_p}\over n^{-1/2}}\rfloor}$ is the number of rings needed to cover the whole network, $\cR_{1\beta}$ is the set of the senders in $R_{1\beta}$, $A(R_{1\beta})$  is the area of $R_{1\beta}$, and {$K_1$ is a positive constant independent with $n$.} Here, $(a)$ is since the senders in the preservation regions do not transmit, $(b)$ is by assuming that the senders in each ring are at the boundary close to warden $w$,  $(c)$ is by using weak law of large numbers (WLLN) , $(d)$ is from Corollary \ref{cor1}, $(e)$ is by upper bounding $A(R_{1\beta})$, and $(f)$ is due to the Riemann sum. Since \eqref{eq:45} holds for arbitrary channel use $u$, $\rho_{wm}$ is upper-bounded as:
\begin{align}
\rho_{wm} \leq K_1 P_{\rm{tx}} n r_p^{2-\alpha}.\label{eq:46}
\end{align}
By \eqref{eq:18} and \eqref{eq:46}, the covertness constraint is satisfied if $K_1 P_{\rm{tx}} n r_p^{2-\alpha} \leq \sqrt{2} N_0 \sqrt{\delta \over l}$, or $P_{\rm{tx}}\leq K'_1 l^{-{1\over2}}n^{-1} r_p^{\alpha-2}$ for a constant $K'_1$ independent with $n$. Since $r_p=\Theta(n^{-(s/2+\epsilon')})$ for {any $\epsilon'>0$}, we conclude that  $P_{\rm{tx}}=\Theta(l^{-{1\over2}}{n^{-({s\over2}(\alpha-2)+1)-\epsilon}})$ for {any $\epsilon>0$} satisfies the covertness constraint.
\end{proof}

\begin{figure}
\centering
\includegraphics[width=0.9\columnwidth]{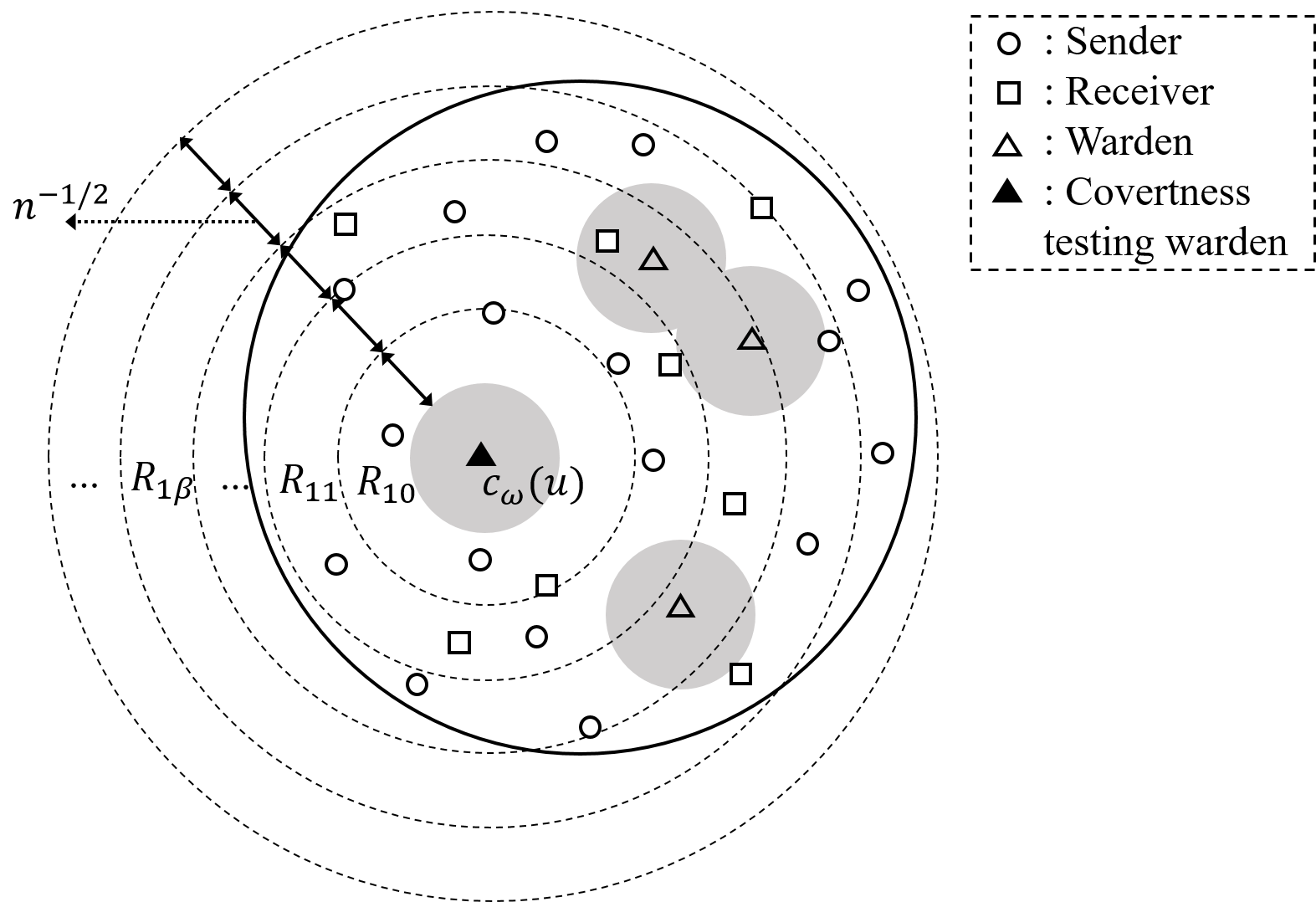}
\caption{A set of disjoint rings to derive an upper bound on the received power at warden $w$ for $0<s<1$. Here, $R_{1\beta}$ denotes the $(\beta+1)$-th smallest ring and each ring has width $n^{-1/2}$.} \label{fig5}
\end{figure}

\begin{figure}
\centering
\includegraphics[width=0.9\columnwidth]{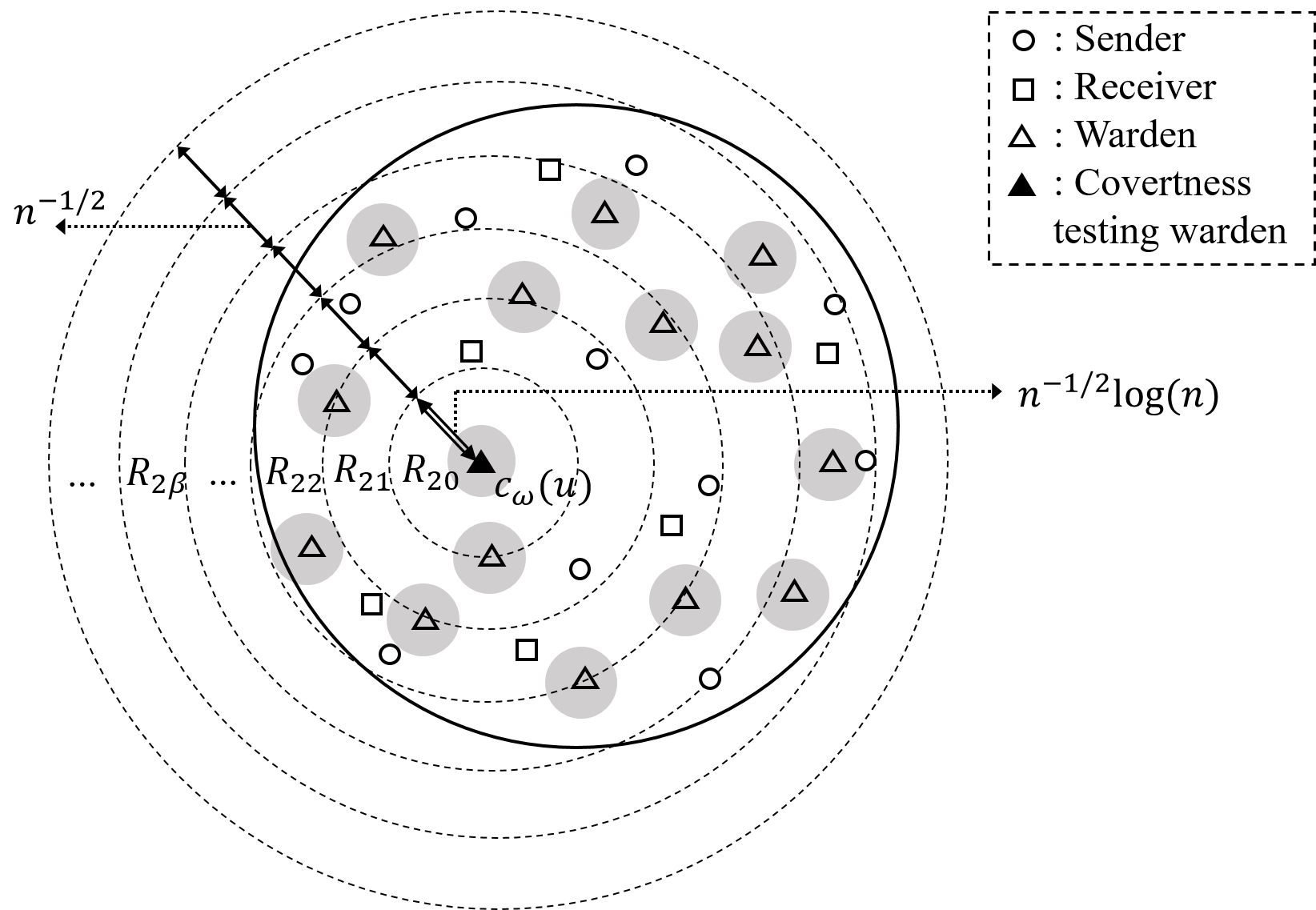} 
\caption{A set of disjoint rings to derive an upper bound on the received power at warden $w$ for $s\geq 1$.  Here, $R_{2\beta}$ denotes the $(\beta+1)$-th smallest ring and each ring except the smallest  has width $n^{-1/2}$. The smallest ring has width $n^{-1/2}\log n-r_p$.}\label{fig6}
\end{figure}

\begin{lemma}\label{lem3}
Let each node transmit  using complex Gaussian coodbook with zero mean and variance of $P_{\rm{tx}}$. If $s\geq1$, then $P_{\rm{tx}}=\Theta(l^{-{1\over2}}{n^{-{s\alpha\over2}-\epsilon}})$ {for any $\epsilon>0$} satisfies the covertness constraint with high probability.
\end{lemma}
\begin{proof}
Let $s_i(u)$ and $c_w(u)$ be the locations of sender $i$ and warden $w$  at the $u$-th channel use the warden observes, respectively. The proof is similar with that of Lemma \ref{lem2}, but now we consider a set of disjoint rings where the smallest ring has width $n^{-1/2}\log n-r_p$ to apply Corollary \ref{cor1}, while the width of the other rings remains as  $n^{-1/2}$ as shown in Fig. \ref{fig6}. Let $R_{2\beta}$ denote the $(\beta+1)$-th smallest ring. Then, $\rho_{w,u}$ is upper-bounded as:
\begin{align}
&\rho_{w,u}\leq \sum_{i:|s_i(u)-c_w(u)|\geq r_p}\!\!\!\!\!P_{\rm{tx}}(i,u)\cdot {G\over{|s_i(u)-c_w(u)|^\alpha}}&& \label{eq:47}\\
\begin{split}\label{eq:48}
&\overset{(a)}\leq\sum_{i\in \cR_{20}}P_{\rm{tx}}(i,u)\cdot {G\over r_p^{\alpha}}+\sum_{\beta=1}^{\kappa_2(n,u)}\sum_{i\in \cR_{2\beta}}P_{\rm{tx}}(i,u)\\
&\cdot {G\over(n^{-1/2}\log n+(\beta-1) n^{-1/2})^\alpha}
\end{split}\\
\begin{split}\label{eq:48.1}
&\leq P_{\rm{tx}}(1+\epsilon'')\cdot {G\over r_p^{\alpha}}\cdot |\cR_{20}|+\sum_{\beta=1}^{\kappa_2(n,u)}P_{\rm{tx}}(1+\epsilon'') \\
&\cdot {G\over(n^{-1/2}\log n+(\beta-1) n^{-1/2})^\alpha}\cdot |\cR_{2\beta}|
\end{split}\\
\begin{split}\label{eq:48.2}
&\leq P_{\rm{tx}}(1+\epsilon'')\cdot {G\over r_p^{\alpha}}\cdot (1+\delta)\theta nA(R_{20})+\sum_{\beta=1}^{\kappa_2(n,u)}P_{\rm{tx}}(1+\epsilon'') \\
&\cdot {G\over(n^{-1/2}\log n+(\beta-1) n^{-1/2})^\alpha}\cdot (1+\delta)\theta nA(R_{2\beta})
\end{split}\\
\begin{split}\label{eq:49}
&\leq P_{\rm{tx}}(1+\epsilon'')\cdot {G\over r_p^{\alpha}}\cdot (1+\delta)\theta n2\pi(n^{-1/2}\log n)^2\\
&+\sum_{\beta=1}^{\kappa_2(n,u)} P_{\rm{tx}}(1+\epsilon'')\cdot {G\over(n^{-1/2}\log n+(\beta-1) n^{-1/2})^\alpha}\\
&\cdot (1+\delta)\theta n 2\pi(n^{-1/2}\log n+\beta n^{-1/2})n^{-1/2}
\end{split}\\
\begin{split}\label{eq:50}
&\leq P_{\rm{tx}}(1+\epsilon'')\cdot {G\over r_p^{\alpha}}\cdot (1+\delta)\theta n2\pi(n^{-1/2}\log n)^2 \\
&+\int_{n^{-1/2}(\log n -1)}^{{1/\sqrt{\pi}}+|c_w(u)|} \!\!\!\!\!\!\!\!\!\! P_{\rm{tx}}(1+\epsilon'')\cdot{G\over x^\alpha}\cdot (1+\delta)\theta n2\pi(x+n^{-1/2})dx
\end{split}\\
&\overset{(b)}\leq K_2 P_{\rm{tx}} (\log n)^2 r_p^{-\alpha},\label{eq:51}
\end{align}
with high probability {for any $\epsilon''>0$ and $\delta>0$,} where  $P_{\rm{tx}}(i,u)$ is the transmission power of sender $i$ at channel use $u$, $\kappa_2(n,u)={\lfloor{{1/\sqrt{\pi}+|c_w(u)|-n^{-1/2}\log n}\over n^{-1/2}}\rfloor}+1$ is the number of rings needed to cover the whole network, $\cR_{2\beta}$ is the set of the senders in $R_{2\beta}$, $A(R_{2\beta})$  is the area of $R_{2\beta}$, {and $K_2$ is a positive constant independent with $n$.} Here, $(a)$ is by separating $R_{20}$ region from the other regions and $(b)$ is since the first term of \eqref{eq:50} is dominant as $n$ goes to infinity. Because \eqref{eq:51} holds for arbitrary channel use $u$, $\rho_{wm}$ is upper-bounded as:
\begin{align}
\rho_{wm} \leq K_2 P_{\rm{tx}} (\log n)^2 r_p^{-\alpha}.\label{eq:52}
\end{align}
By \eqref{eq:18} and \eqref{eq:52}, the covertness constraint is satisfied if $P_{\rm{tx}}\leq K'_2 l^{-{1\over2}} (\log n)^{-2} r_p^{\alpha}$ where  $K'_2$ is a constant   independent with $n$. Since $r_p=\Theta(n^{-(s/2+\epsilon')})$ for {any $\epsilon'>0$}, we conclude that $P_{\rm{tx}}=\Theta(l^{-{1\over2}}{n^{-{s\alpha\over2}-\epsilon}})$ for {any $\epsilon>0$} satisfies the covertness constraint.
\end{proof}

\subsection{Aggregate Throughput}\label{sec5D}
In this subsection, we derive the achievable aggregate throughput in Theorems \ref{thm1} and \ref{thm2}. We first derive a feasible long-term throughput  for each sender-receiver pair, which we call pairwise throughput in short. A pairwise throughput $R_{\rm{pair}}(n,s)$  is said to be feasible if  
\begin{align}
\limit_{T\rightarrow\infty} {1\over T}{\sum_{t=1}^{T}R_{j\rm{th-pair}}(n,s,t)} \geq R_{\rm{pair}}(n,s), ~\forall j \label{eq:55}
\end{align}
 where $R_{j\rm{th-pair}}(n,s,t)$ is the throughput of the $j$th sender-receiver pair at time $t$. The following lemma shows a relationship between the aggregate throughput and the pairwise throughput. 

\begin{lemma}\label{lem4}
Let $R_{\rm{pair}}(n,s)$ be a feasible pairwise throughput. Then an aggregate throughput of $T(n,s)=\Theta(n) \cdot R_{\rm{pair}}(n,s)$ is achievable with high probability.
\end{lemma}
\begin{proof}
As mentioned in Section \ref{sec5B}, we assume the steady state and hence assume that each sender has a data packet destined for its receiver in phase 2. Since the senders are uniformly and randomly chosen for each phase 1 and phase 2 and we are considering the throughputs in the long-term sense, the aggregate throughput equals to the half (because phase 2 occupies the half of the time) of the product of the number of senders in phase 2 and the pairwise throughput. The proof is completed by noting that  the number of senders in phase 2 is between $((1-\delta)\theta n(1-\epsilon(n_w, r_p)),(1+\delta)\theta n(1-\epsilon(n_w, r_p)))$ for any $\delta>0$, as $n$ goes to infinity by Corollary \ref{cor1}. 
\end{proof}
Now, we derive a feasible pairwise throughput. Since the nodes use the Gaussian codebook as described in Section \ref{sec5B}, $R_{j\rm{th-pair}}(n,s,t)$ is represented as:
\begin{align}
R_{j\rm{th-pair}}(n,s,t)=\log\left(1+{{P_{\rm{tx}}\cdot {r_j(t)}^{-\alpha}}\over{N_0+I_j(n,s,t)}} \right), \label{eq:56}
\end{align}
where $r_j(t)$ is the distance between the $j$th sender-receiver pair and $I_j(n,s,t)$ is the interference power at the $j$th receiver at time $t$. To derive $R_{\rm{pair}}(n,s)$, we use the following lemmas on the distribution of $r_j(t)$ and an upper bound on $I_j(n,s,t)$.

\begin{lemma}\label{lem5}
Let $F(z)$ be the cumulative distribution of the distance $z$ between a sender-receiver pair at time $t$. For $z=\Theta(n^{-\epsilon})$, $F(z)=1-\exp(-\pi z^2n(1-\theta))$ for {any $\epsilon>0$}.
\end{lemma}
\begin{proof}
 For an arbitrary sender-receiver pair, let $s(t)$ and $v(t)$ be the location of the sender and the receiver at time $t$, respectively. Then,
\begin{align}
\begin{split}\label{eq:57}
F(z)&=p(|s(t)-v(t)|\leq z)\\
&=p(\min_{i\in\cV_t}(|s(t)-v_i(t)|)\leq z)
\end{split}\\
&\overset{(a)}=1-\limit_{n\rightarrow\infty}\prod_{i=1}^{n(1-\theta)}p(|s(t)-v_i(t)|> z)\label{eq:58}\\
&=1-\limit_{n\rightarrow\infty}(1-\pi z^2)^{n(1-\theta)}\label{eq:59}\\
&=1-\limit_{n\rightarrow\infty}(1-\pi z^2)^{{1\over{\pi z^2}}\cdot \pi z^2 n(1-\theta)}\label{eq:59.1}\\
&\overset{(b)}=1-\exp(-\pi z^2n(1-\theta)),\label{eq:60}
\end{align}
where $v_i(t)$ is the location of receiver $i$ at time $t$ and $\cV_t$ is the set of receivers at time $t$. Here, $(a)$ is since $n$ nodes are uniformly and randomly distributed and $(b)$ is because $\limit_{x\rightarrow 0}(1-x)^{1/x}\triangleq e^{-1}$ by the definition of Euler's number.
\end{proof}

The following two corollaries, which can be proved similarly with Lemma \ref{lem5}, are used to derive an upper bound on the interference power at a receiver in Lemma \ref{lem6} and an upper bound on the aggregate throughput in Section \ref{sec6}.
\begin{corollary}\label{cor2}
Let $F_r(z)$ be the cumulative distribution of the distance $z=\Theta(n^{-\epsilon})$  between a receiver and the nearest sender from the receiver at time $t$. Then, $F_r(z)=1-\exp(-\pi z^2n\theta)$ for {any $\epsilon>0$.}
\end{corollary}
\begin{corollary}\label{cor3}
Let $F_s(z)$ be the cumulative distribution of the distance $z=\Theta(n^{-\epsilon})$  between a node and the nearest node from the node at time $t$. Then, $F_s(z)=1-\exp(-\pi z^2n)$ for {any $\epsilon>0$.}
\end{corollary}

\begin{figure}
\centering
\includegraphics[width=0.9\columnwidth]{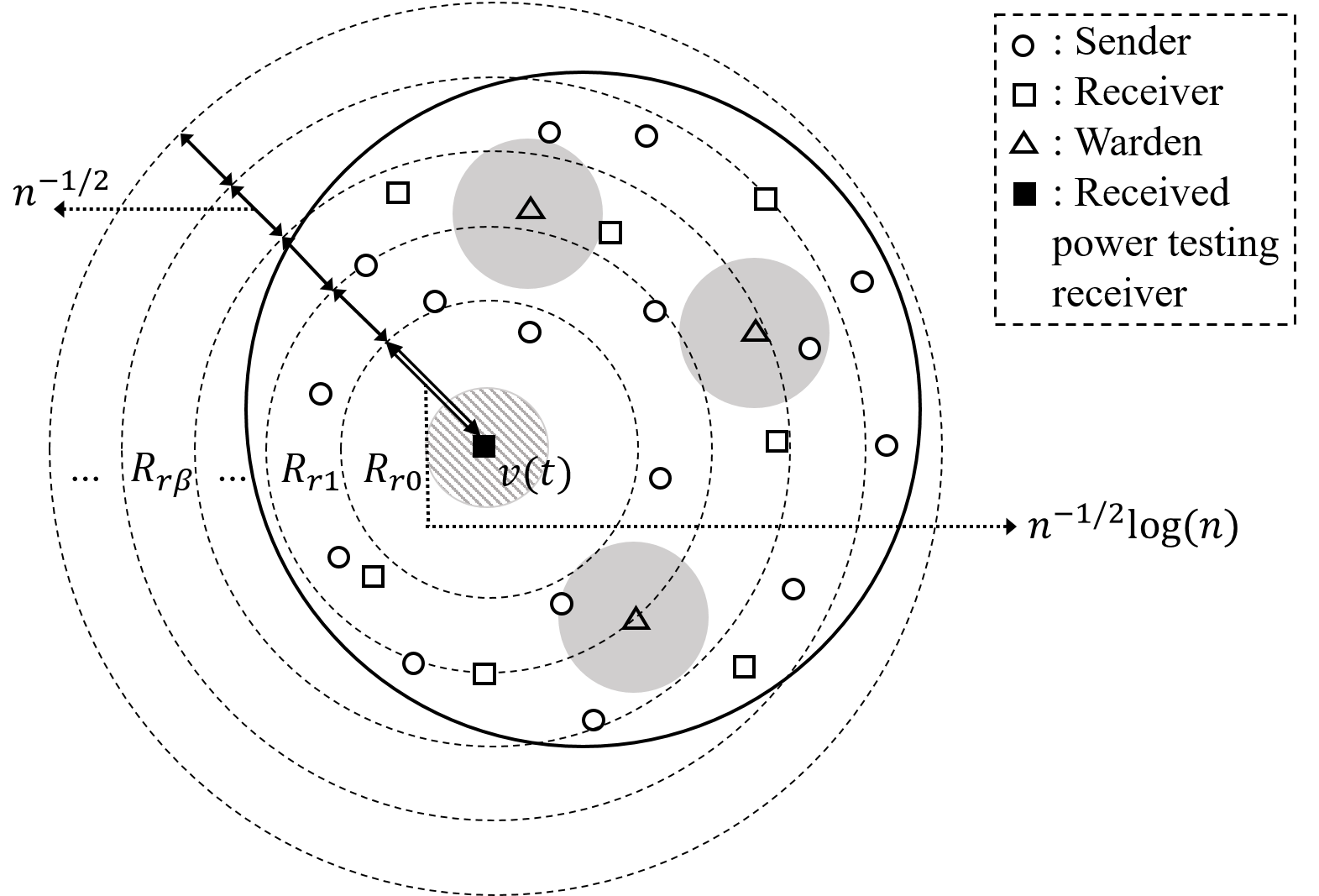} 
\caption{A set of disjoint rings to derive an upper bound on the interference power at a receiver. Here, the (dashed) disk of radius $n^{-({1\over2}+\epsilon')}$ {for an arbitrarily small $\epsilon'>0$} denotes the region where there is no sender in the disk with high probability, $R_{r\beta}$ denotes the $(\beta+1)$-th smallest ring, and each ring except the smallest  has width $n^{-1/2}$. The smallest ring has width $n^{-1/2}\log n-n^{-({1\over2}+\epsilon')}$.}\label{fig7}
\end{figure}

\begin{lemma}\label{lem6}
Let $I(n,s,t)$ be the  interference power at a receiver at time $t$. Then  $I(n,s,t)<P_{\rm{tx}}n^{{\alpha\over2}+\epsilon}$ with high probability for {any $\epsilon>0$}.
\end{lemma}
\begin{proof}
Consider an arbitrary receiver in the network and let $v(t)$ denote its location at time $t$. The proof is similar with that of Lemma 3, but now we consider a set of disjoint rings centered at the receiver which cover the whole network except the disk centered at the receiver where there is no sender with high probability. The radius of the disk is $n^{-({1\over2}+\epsilon')}$ {for an arbitrarily small $\epsilon'>0$} by Corollary \ref{cor2} since $F_r(n^{-({1\over2}+\epsilon')})$ converges to zero as $n$ goes to infinity for any $\epsilon'>0$. The smallest ring has width $n^{-1/2}\log n-n^{-({1\over2}+\epsilon')}$ and the other rings have width $n^{-1/2}$ to apply Corollary \ref{cor1}.  Let $I(n,s,t)$ be the interference power at the receiver at time $t$ and $R_{r\beta}$ denote the $(\beta+1)$-th smallest ring. Then, $I(n,s,t)$ is upper-bounded as:
\begin{align}
I(n,s,t)&\overset{(a)}\leq\!\!\!\!\!\!\!\!\!\! \sum_{i:|s_i(t)-v(t)|\geq n^{-({1\over2}+\epsilon')}}\!\!\!\!\!\! P_{\rm{tx}}\cdot {G\over{|s_i(t)-v(t)|^\alpha}}\label{eq:61}\\
\begin{split}\label{eq:62}
&\overset{(b)}\leq \sum_{i\in \cR_{r0}}P_{\rm{tx}}\cdot {G\over  n^{-\alpha({1\over2}+\epsilon')}}+\sum_{\beta=1}^{\kappa_r(n,t)}\sum_{i\in \cR_{r\beta}}P_{\rm{tx}}\\
&\cdot {G\over(n^{-1/2}\log n+(\beta-1) n^{-1/2})^\alpha}
\end{split}\\
&\overset{(c)}\leq K_r P_{\rm{tx}} (\log n)^2 n^{\alpha({1\over2}+\epsilon')}\label{eq:66}\\
&< P_{\rm{tx}}n^{{\alpha\over2}+\epsilon},\label{eq:67}
\end{align}
with high probability {for any $\epsilon>0$ and $\epsilon'>0$}, where $s_i(t)$ denote the location of the $i$-th sender at time $t$, $\kappa_r(n,t)={\lfloor{{1/\sqrt{\pi}+|v(t)|-n^{-1/2}\log n}\over n^{-1/2}}\rfloor}+1$ is the number of rings needed to cover the whole network, $\cR_{r\beta}$ is the set of the senders in $R_{r\beta}$, {and $K_r$ is a positive constant independent with $n$.} Here, $(a)$ is since the senders in the preservation region do not transmit, $(b)$ is by separating $A_{r0}$ region from the other regions, and $(c)$ can be proved  similarly as in \eqref{eq:48.1}-\eqref{eq:50} in the proof of Lemma~\ref{lem3}.
\end{proof}

By using Lemmas \ref{lem5} and \ref{lem6}, we derive a feasible pairwise throughput as follows:
\begin{align}
\begin{split}\label{eq:68}
&\limit_{T\rightarrow\infty}{1\over T}{\sum_{t=1}^{T}R_{j\rm{th-pair}}{(n,s,t)}} \\
&=\limit_{T\rightarrow\infty} {1\over T}{\sum_{t=1}^{T}\log\left(1+{{P_{\rm{tx}}\cdot{r_j(t)}^{-\alpha}}\over{N_0+I_j(n,s,t)}} \right)}
\end{split}\\
&\overset{(a)}\geq{\EE}_t\left(\log\left(1+{{P_{\rm{tx}}\cdot {r_j(t)}^{-\alpha}}\over{N_0+I_j(n,s,t)}} \right)\right)(1-\epsilon'')\label{eq:69}\\ 
\begin{split}\label{eq:70}
&\geq p(r_j(t)< \min(P_{\rm{tx}}^{1\over\alpha},n^{-1/2}))\\
&\cdot {\EE}_t\left(\log\left(1+{{P_{\rm{tx}}\cdot {r_j(t)}^{-\alpha}}\over{N_0+I_j(n,s,t)}} \right) \right. \\
&\ \ \ \ \ \ \ \ \ \  \left. \bigg\rvert r_j(t)< \min(P_{\rm{tx}}^{1\over\alpha},n^{-1/2}) \right)(1-\epsilon'')
\end{split}\\
&\overset{(b)}\geq\begin{cases}\label{eq:71}
F(P_{\rm{tx}}^{1\over\alpha})\cdot\log\left(1+{{1}\over{N_0+1}} \right)(1-\epsilon'') \\
\ \ \ \ \ \ \ \ \ \ \ \ \mbox{for} \ \ P_{\rm{tx}}^{1\over\alpha}\leq n^{-1/2},\\
F( n^{-1/2}) \cdot \log\left(1+{{P_{\rm{tx}}\cdot}n^{\alpha\over2}\over{N_0+P_{\rm{tx}}n^{{\alpha\over2}+\epsilon}}} \right)(1-\epsilon'')\\
\ \ \ \ \ \ \ \ \ \ \ \ \mbox{for} \ \ P_{\rm{tx}}^{1\over\alpha}> n^{-1/2},
\end{cases}\\
&\overset{(c)}\geq F(\min(P_{\rm{tx}}^{1\over\alpha},n^{-1/2}))\cdot\log\left(1+{{1}\over{N_0+1}}\right)(1-\epsilon'')\label{eq:72}\\
\begin{split}\label{eq:73}
&\overset{(d)}\geq \min\left({{\pi P_{\rm{tx}}^{2/\alpha}n\theta}\over{1+{\pi P_{\rm{tx}}^{2/\alpha}n\theta}}},{{\pi\theta}\over{1+{\pi\theta}}}\right)\\
&\cdot \log\left(1+{{1}\over{N_0+1}} \right)(1-\epsilon'')
\end{split}\\
&\geq \min( K'_1n^{1-\epsilon'}P_{\rm{tx}}^{2/\alpha},K'_2n^{-\epsilon'}),\label{eq:74}
\end{align}
with high probability for arbitrary $j$ and {any $\epsilon'>0$ and $\epsilon''>0$}, where {$K'_1$ and $K'_2$ are positive constants independent with $n$.} Here, $(a)$ is by WLLN, $(b)$ is by Lemma \ref{lem6}, $(c)$ is since ${x}\over{N_0+x}$ is an increasing function for $x>0$, and $(d)$ follows from Lemma \ref{lem5} and $1-e^{-x}>{x/{(x+1)}}$ for $x>0$.


Since $R_{\rm{pair}}(n,s)=\min( K'_1n^{1-\epsilon'}P_{\rm{tx}}^{2/\alpha},K'_2n^{-\epsilon'})$ is feasible for any $\epsilon'>0$, by Lemma \ref{lem4}, the following aggregate throughput is achievable for $\epsilon'>0$:
\begin{align}
T(n,s)=n\cdot\min( K'_1n^{1-\epsilon'}P_{\rm{tx}}^{2/\alpha},K'_2n^{-\epsilon'}).\label{eq:76}
\end{align}

Now, Theorems \ref{thm1} and \ref{thm2} are proved by substituting $ P_{\rm{tx}}$ in Lemmas \ref{lem2} and \ref{lem3} into \eqref{eq:76}, respectively.



\section{Converse}\label{sec6}
In this section, we prove Theorems \ref{thm_ub}, \ref{thm3}, and \ref{thm4}. We note that the proofs do not depend on whether the wardens have mobility or not.

\subsection{Proof of Theorem  \ref{thm_ub}}\label{sec6A}
In this subsection, we derive an upper bound on the aggregate throughput by assuming there is no covertness constraint. Let source node $j$ communicate to its destination node $k_j$. Then, $T(n,s)$ is upper-bounded as:
\begin{align}
&T(n,s)= \limit_{T\rightarrow\infty}{1\over T}\sum_{t=1}^{T}\sum_{j=1}^{n}R_{jk_j}(n,s,t)\label{eq:83.1}\\
&\overset{(a)}\leq\limit_{T\rightarrow\infty}{1\over T}\sum_{t=1}^{T}\sum_{j=1}^{n} \log\left(1+\frac{P}{N_0}\sum_{\scriptstyle {i=1}\atop \scriptstyle {i\neq j}}^{n}{G\over{d_{ji}(t)^\alpha}}\right)\label{eq:83.2}\\
&\overset{(b)}\leq\limit_{T\rightarrow\infty}{1\over T}\sum_{t=1}^{T}\sum_{j=1}^{n}\log\left(1+\frac{P}{N_0}(n-1) Gn^{\alpha(1+\epsilon')}\right)\label{eq:83.3}\\
&\leq n\cdot\log\left(1+{{{P}}\over N_0}\cdot Gn^{\alpha(1+\epsilon')+1}\right)\label{eq:86}\\
&\leq K_{\rm{tr}}n^{1+\epsilon},\label{eq:87}
\end{align}
with high probability {for any $\epsilon>0$ and $\epsilon'>0$}, where $d_{ji}(t)$ is the distance between nodes $j$ and $i$ at time $t$ and {$K_{\rm{tr}}$ is a positive constant independent with $n$.} Here $(a)$ is since $R_{jk_j}(n,s,t)$ is upper-bounded by the capacity of the single input multiple output (SIMO) channel between node $j$ and all the other nodes and $(b)$ is since the probability that the minimum distance between two nodes in the network is smaller than $n^{-(1+\epsilon')}$ converges to zero {for any $\epsilon'>0$} as $n$ goes to infinity, i.e., 
\begin{align}
\begin{split}\label{eq:88}
p&\left(\min_{(j,i),j\neq i}d_{ji}(t)<n^{-(1+\epsilon')}\right) \\
&\leq \sum_{j=1}^{n}p\left(\min_{i\neq j}d_{ji}(t)<n^{-(1+\epsilon')}\right)
\end{split}\\
&=n\cdot p\left(\min_{i \geq 2 }d_{1i}(t)<n^{-(1+\epsilon')}\right)\label{eq:88.1}\\
&\overset{(a)}=n\left(1-\left(1-{\pi\over n^{2+2\epsilon'}}\right)^{n-1}\right),\label{eq:89}\\
&\overset{n\rightarrow \infty} \longrightarrow 0 
\end{align}
where $(a)$ is from Collorary \ref{cor3}. This completes the proof. 

\subsection{Proof of Theorems \ref{thm3} and \ref{thm4}}\label{sec6B}
Assume that each node not contained in the regions of radius $\Theta(n^{-(\frac{s}{2}+\epsilon')})$  around each warden for an arbitrarily small $\epsilon'>0$ uses the same power at each channel use. These regions correspond to the preservation regions introduced in the proposed scheme for proving achievability. We note that we do not restrict the behavior of the nodes inside these regions in the converse proof, while we force them not to transmit in the achievability proof.  

We start by deriving an upper bound on the transmission power satisfying \eqref{eq:32}. The proof is similar with that of Lemmas \ref{lem2} and \ref{lem3}, but now we derive a lower bound on the received power at a warden instead of an upper bound. To derive the lower bound, we consider a set of disjoint rings, centered at the warden, that is covered by the whole network as shown in Figs. \ref{fig8} and \ref{fig9}. The following lemmas present upper bounds on the transmission power satisfying the covertness constraint for $0<s<1$ and $s\geq1$.

\begin{figure}
\centering
\includegraphics[width=0.9\columnwidth]{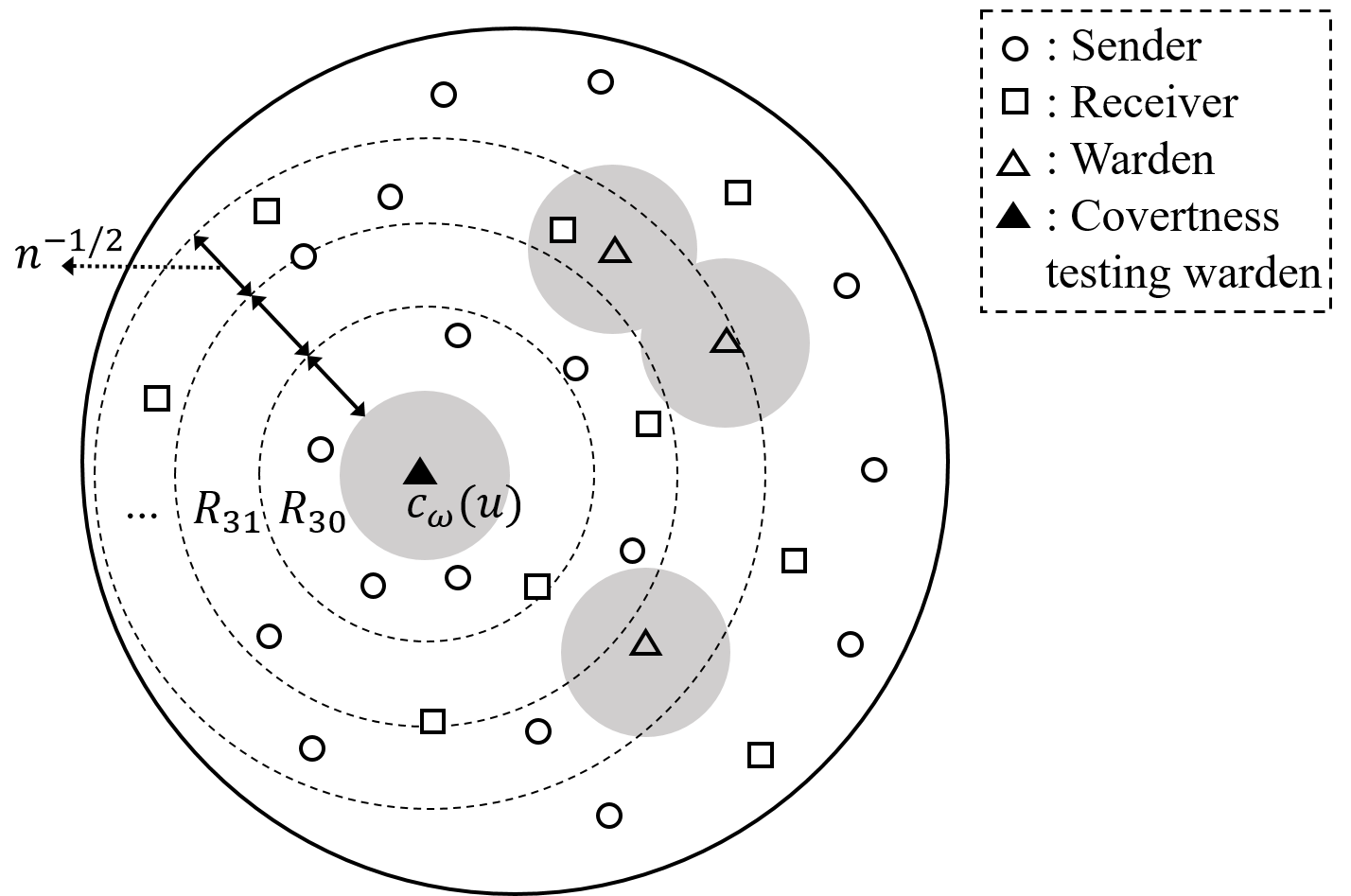}
\caption{A set of disjoint rings to derive a lower bound on the received power at warden $w$ for $0<s<1$. Here, $R_{3\beta}$ denotes the $(\beta+1)$-th smallest ring and each ring has width $n^{-1/2}$.} \label{fig8}
\end{figure}

\begin{lemma}\label{lem7}
Let each node not contained in the regions of radius $r'_p=\Theta(n^{-(\frac{s}{2}+\epsilon')})$ around each warden for an arbitrarily small $\epsilon'>0$ transmit the same power of $P_{\rm{tx}}$ at each channel use. If $0<s<1$ and the network satisfies the covertness constraint, then $P_{\rm{tx}}\leq K_3 l^{-{1\over2}}{n^{-({s\over2}(\alpha-2)+1)+\epsilon}}$ {for any $\epsilon>0$} and a positive constant $K_3$ independent with $n$ with high probability.
\end{lemma}
\begin{proof}
Let $s_i(u)$ and $c_w(u)$ be the locations of node $i$ and warden $w$ at channel use $u$, respectively.The proof is similar with that of Lemma \ref{lem2}, but now we consider the largest set of disjoint rings that is covered by the whole network as shown in Fig. \ref{fig8}. Let $R_{3\beta}$ denote the $(\beta+1)$-th smallest ring. Then, the received power at warden $w$ at channel use $u$, $\rho_{w,u}$, is lower-bounded as:
\begin{align}
\rho_{w,u}&\geq\sum_{i:s_i(u)\notin\cP_u}P_{\rm{tx}}\cdot {G\over{|s_i(u)-c_w(u)|^\alpha}}\label{eq:91}\\
&\overset{(a)}\geq{1\over2}\cdot\!\!\sum_{i:|s_i(u)-c_w(u)|\geq r'_p}P_{\rm{tx}}\cdot {G\over{|s_i(u)-c_w(u)|^\alpha}}\label{eq:91.1}\\
&\overset{(b)}\geq{1\over2}\cdot\sum_{i:1/\sqrt{\pi}-|c_w(u)|\geq|s_i(u)-c_w(u)|\geq r'_p}\!\!\!\!\!\!\!\!\!\!P_{\rm{tx}}\cdot {G\over{|s_i(u)-c_w(u)|^\alpha}}\label{eq:92}\\
&\overset{(c)}\geq{1\over2}\cdot\sum_{\beta=0}^{\kappa_3(n,u)}P_{\rm{tx}}\cdot {G\cdot |\cR_{3\beta}|\over(r'_p+(\beta+1)n^{-1/2})^\alpha}\label{eq:93}\\
&\overset{(d)}\geq{1\over2}\cdot\sum_{\beta=0}^{\kappa_3(n,u)}P_{\rm{tx}}\cdot {G\cdot(1-\delta)nA(R_{3\beta})\over(r'_p+(\beta+1)n^{-1/2})^\alpha}\label{eq:94}\\
\begin{split}\label{eq:95}
&\overset{(e)}\geq{1\over2}\cdot\sum_{\beta=0}^{\kappa_3(n,u)}P_{\rm{tx}}\cdot {G\over(r'_p+(\beta+1)n^{-1/2})^\alpha}\\
&\cdot (1-\delta)n 2\pi(r'_p+\beta n^{-1/2})n^{-1/2}
\end{split}\\
&\geq\int_{r'_p+n^{-1/2}}^{{1/\sqrt{\pi}}-|c_w(u)|}\!\!\!\!\!\!\!\!P_{\rm{tx}}\cdot{G\over x^\alpha}\cdot (1-\delta)n\pi(x-n^{-1/2})dx\label{eq:96}\\
&\geq K'_3 P_{\rm{tx}} n {r'_p}^{2-\alpha},\label{eq:97}
\end{align}
with high probability {for any $\delta>0$}, where $\cP_u$ is the set of the nodes contained in the regions of radius $r'_p$ around each warden at channel use $u$, $\kappa_3(n,u)={\lfloor{{1/\sqrt{\pi}-|c_w(u)|-r'_p}\over n^{-1/2}}\rfloor}-1$ is the number of the maximum rings covered by the whole network, $\cR_{3\beta}$ is the set of the nodes in $R_{3\beta}$, $A(R_{3\beta})$ is the area of $R_{3\beta}$, {and $K'_3$ is a positive constant independent with $n$.} Here, $(a)$ is since the total area of the regions of radius $r'_p$ around each warden is smaller than $1/2$, $(b)$ is because we only consider the rings inside the network, $(c)$ is by assuming that the nodes in each ring are at the boundary far from warden $w$, $(d)$ is from Corollary \ref{cor1}, and $(e)$ is by lower bounding $A_{3\beta}$. Since \eqref{eq:97} holds for arbitrary channel use $u$, $\bar\rho_w$ is lower-bounded as
\begin{align}
\bar\rho_w \geq K'_3 P_{\rm{tx}} n {r'_p}^{2-\alpha}.\label{eq:98}
\end{align}
By \eqref{eq:32} and \eqref{eq:98}, if the covertness constraint is satisfied, then $ K'_3 P_{\rm{tx}} n {r'_p}^{2-\alpha} \leq \sqrt{2} N_0 \sqrt{\delta \over l}+o(l^{-1/2})$, or $P_{\rm{tx}}\leq K''_3 l^{-{1\over2}}n^{-1} {r'_p}^{\alpha-2}$ for a positive constant $K''_3$ independent with $n$. Since $r'_p=\Theta(n^{-(s/2+\epsilon')})$ {for an arbitrarily small $\epsilon'>0$,} we conclude that  $P_{\rm{tx}}\leq K_3 l^{-{1\over2}}{n^{-({s\over2}(\alpha-2)+1)+\epsilon}}$ {for any $\epsilon>0$} satisfies the covertness constraint for a positive constant $K_3$ independent with $n$.
\end{proof}

\begin{figure}
\centering
\includegraphics[width=0.9\columnwidth]{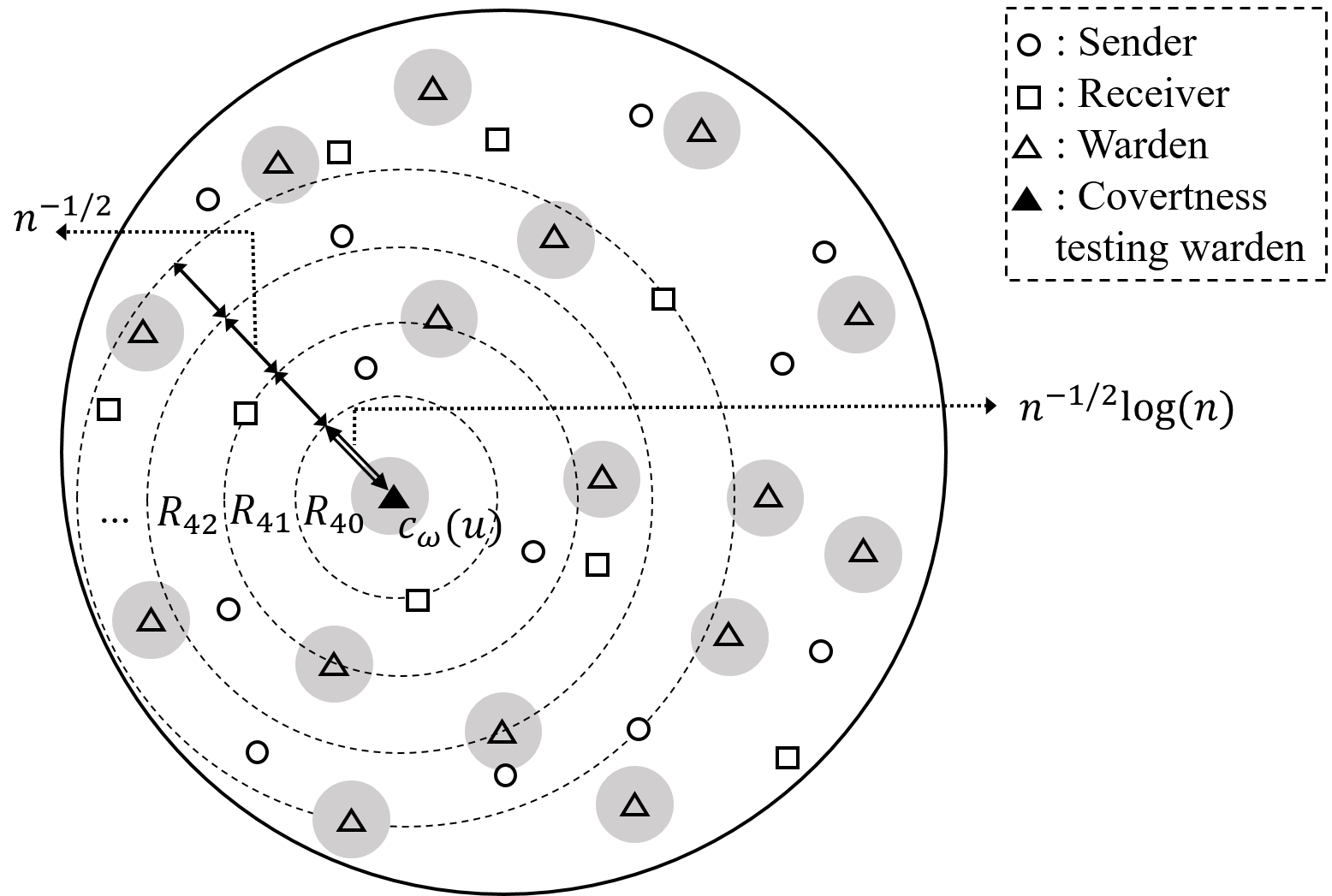} 
\caption{A set of disjoint rings to derive a lower bound on the received power at warden $w$ for $s\geq1$.  Here, $R_{4\beta}$ denotes the $(\beta+1)$-th smallest ring and each ring except the smallest has width $n^{-1/2}$. The smallest ring has width $n^{-1/2}\log n-r'_p$.}\label{fig9}
\end{figure}

\begin{lemma}\label{lem8}
Let each node not contained in the regions of radius $r'_p=\Theta(n^{-(\frac{s}{2}+\epsilon')})$  around each warden for an arbitrarily small $\epsilon'>0$ transmit the same power of $P_{\rm{tx}}$ at each channel use. If $s\geq1$ and the network satisfies the covertness constraint, then $P_{\rm{tx}}\leq K_4 l^{-{1\over2}}{n^{-{\alpha\over2}-\epsilon}}$ {for any $\epsilon>0$ and a positive constant $K_4$} independent with $n$ with high probability.
\end{lemma}
\begin{proof}
Let $s_i(u)$ and $c_w(u)$ be the locations of node $i$ and warden $w$ at channel use $u$, respectively. The proof is similar with that of Lemma \ref{lem7}, but now we consider a set of disjoint rings where the smallest ring has width $n^{-1/2}\log n-r'_p$ to apply Corollary \ref{cor1}, while the width of the other rings remains as  $n^{-1/2}$ as shown in Fig. \ref{fig9}. Let  $R_{4\beta}$ denote the $(\beta+1)$-th smallest ring. Then, the received power at warden $w$ at channel use $u$, $\rho_{w,u}$ is lower-bounded as:
\begin{align}
&\rho_{w,u}\geq\sum_{i:s_i(u)\notin\cP_u}P_{\rm{tx}}\cdot {G\over{|s_i(u)-c_w(u)|^\alpha}}\label{eq:99}\\
&\geq{1\over2}\cdot\sum_{i:1/\sqrt{\pi}-|c_w(u)|\geq|s_i(u)-c_w(u)|\geq r'_p}\!\!\!\!\!\!\!\!\!\!\!P_{\rm{tx}}\cdot {G\over{|s_i(u)-c_w(u)|^\alpha}}\label{eq:100}\\
\begin{split}\label{eq:101}
&\overset{(a)}\geq{1\over2}\cdot P_{\rm{tx}}\cdot {G\over(n^{-1/2}\log n)^{\alpha}}\cdot |\cR_{40}|\\
&+{1\over2}\cdot\sum_{\beta=1}^{\kappa_4(n,u)}P_{\rm{tx}}\cdot {G\over(n^{-1/2}\log n+\beta n^{-1/2})^\alpha}\cdot |\cR_{4\beta}|
\end{split}\\
&\overset{(b)}\geq K'_4 P_{\rm{tx}}n^{\alpha/2}(\log n)^{2-\alpha},\label{eq:104}
\end{align}
with high probability, where $\cP_u$ is the set of the nodes contained in the regions of radius $r'_p$ around each warden at channel use $u$, $\kappa_4(n,u)={\lfloor{{1/\sqrt{\pi}-|c_w(u)|-n^{-1/2}\log n}\over n^{-1/2}}\rfloor}$ is the number of rings covered by the whole network, $\cR_{4\beta}$ is the set of the nodes in $R_{4\beta}$, and $K'_4$ is a positive constant independent with $n$. Here, $(a)$ is by separating $A_{40}$ region from the other regions and $(b)$ can be proved similarly as in \eqref{eq:94}-\eqref{eq:96} in the proof of  Lemma \ref{lem7}. Because \eqref{eq:104} holds for arbitrary time $t$, $\bar\rho_w$ is lower-bounded as:
\begin{align}
\bar\rho_w \geq K'_4 P_{\rm{tx}}n^{\alpha/2}(\log n)^{2-\alpha}.\label{eq:105}
\end{align}
By \eqref{eq:32} and \eqref{eq:105}, if the covertness constraint is satisfied, then $ K'_4 P_{\rm{tx}}n^{\alpha/2}(\log n)^{2-\alpha} \leq \sqrt{2} N_0 \sqrt{\delta \over l}+o(l^{-1/2})$ and we conclude that $P_{\rm{tx}}\leq K_4 l^{-{1\over2}}n^{{\alpha\over2}+\epsilon}$ for a positive constant $K_4$ independent with $n$ and { any $\epsilon>0$.}
\end{proof}

On the other hand, the following lemma presents an upper bound on the throughput of an arbitrary source-destination pair in terms of the distance between the source and its nearest node.
\begin{lemma}\label{lem9}
Let $R(n,s,t)$ be the throughput of an arbitrary source-destination pair and $d_s(t)$ be the distance between the source and the nearest node from the source at time $t$. Then, $R(n,s,t)\leq\log\left(1+{P_{\rm{tx}} \over N_0}\cdot{{Kn^{\epsilon}}\over{d_s(t)^\alpha}}\right)$ {for any $\epsilon>0$} with high probability.
\end{lemma}
\begin{proof}
Let us consider source node $j$. Then, $R(n,s,t)$ is upper-bounded as:
\begin{align}
R(n,s,t)&\overset{(a)}\leq\log\left(1+{P_{\rm{tx}} \over N_0}\sum_{\scriptstyle {i=1}\atop \scriptstyle {i\neq j}}^{n}{G\over{d_{ji}(t)^\alpha}}\right)\label{eq:107}\\
&\overset{(b)}\leq \log\left(1+{P_{\rm{tx}} \over N_0}\cdot{{Kn^{\epsilon}}\over{d_s(t)^\alpha}}\right),\label{eq:108}
\end{align}
with high probability for {any $\epsilon>0$}, where $d_{ji}(t)$ is the distance between nodes $j$ and $i$ at time $t$ {and $K$ is a positive constant independent with $n$.} Here $(a)$ is since $R(n,s,t)$ is upper-bounded by the throughput of single input multiple output (SIMO) channel between node $j$ and the other nodes and $(b)$ is proved similar to the proof of Lemma \ref{lem6}, but now there is no node in the disk with radius $d_s(t)$ centered at source $j$. 
\end{proof}

Now we are ready to derive an upper bound on the aggregate   throughput $T(n,s)$. Let source node $j$ communicate to its destination node $k_j$. Then, 
\begin{align}
T(n,s)&=\limit_{T\rightarrow\infty}{1\over T}\sum_{t=1}^{T}\sum_{j=1}^{n}R_{jk_j}(n,s,t)\label{eq:109}\\
&\overset{(a)}\leq\limit_{T\rightarrow\infty}{1\over T}\sum_{t=1}^{T}\sum_{j=1}^{n}\log\left(1+{P_{\rm{tx}} \over N_0}\cdot{{Kn^{\delta}}\over{d_{s,j}(t)^\alpha}}\right)\label{eq:110}\\
&\overset{(b)}\leq{\EE}_t\left(\sum_{j=1}^{n}\log\left(1+{P_{\rm{tx}} \over N_0}\cdot{{Kn^{\delta}}\over{d_{s,j}(t)^\alpha}}\right)\right)(1+\epsilon'') \label{eq:111}\\ 
&\overset{(c)}= n\cdot{\EE}_t\left(\log\left(1+{P_{\rm{tx}} \over N_0}\cdot{{Kn^{\delta}}\over{d_s(t)^\alpha}}\right)\right)(1+\epsilon''),\label{eq:112}\\
\begin{split}\label{eq:114}
&= n\cdot p(d_s(t)<P_{\rm{tx}}^{1\over\alpha})\\
&\cdot{\EE}_t\left(\log\left(1+{P_{\rm{tx}} \over N_0}\cdot{{Kn^{\delta}}\over{d_s(t)^\alpha}}\right)\bigg\rvert d_s(t)<P_{\rm{tx}}^{1\over\alpha}\right)(1+\epsilon'')\\
&+n\cdot p(d_s(t)\geq P_{\rm{tx}}^{1\over\alpha})\\
&\cdot{\EE}_t\left(\log\left(1+{P_{\rm{tx}} \over N_0}\cdot{{Kn^{\delta}}\over{d_s(t)^\alpha}}\right)\bigg\rvert d_s(t)\geq P_{\rm{tx}}^{1\over\alpha}\right)(1+\epsilon'')
\end{split}\\
\begin{split}\label{eq:115}
&\overset{(d)}\leq K_5\cdot n\cdot p(d_s(t)<P_{\rm{tx}}^{1\over\alpha})\\
&\cdot{\EE}_t\left(\log\left(1+{P_{\rm{tx}} \over N_0}\cdot{{Kn^{\delta}}\over{d_s(t)^\alpha}}\right)\bigg\rvert d_s(t)<P_{\rm{tx}}^{1\over\alpha}\right)
\end{split}\\
&\leq K_5\cdot n\cdot F_s(P_{\rm{tx}}^{1\over\alpha})\cdot n^{\epsilon'}\label{eq:116}\\
&\overset{(e)}\leq K_5\cdot n\cdot \min(\pi nP_{\rm{tx}}^{2/\alpha},1) \cdot n^{\epsilon'}\label{eq:117}\\
&\leq K_5 n^{1+\epsilon'}\min(\pi nP_{\rm{tx}}^{2/\alpha},1) ,\label{eq:118}
\end{align}
with high probability {for any $\epsilon'>0$, $\epsilon''>0$, and $\delta>0$}, {where $K$ and $K_5$ are positive constants independent with $n$,} $d_{s,j}(t)$ is the distance between node $j$ and the nearest node from node $j$ at time $t$, and $d_s(t)=d_{s,1}(t)$ is the distance between node $1$ and the nearest node from node $1$ at time $t$. Here, $(a)$ is by Lemma \ref{lem9}, $(b)$ is by WLLN, $(c)$ is because the nodes are i.i.d., $(d)$ is since the first term of \eqref{eq:114} is dominant, {and $(e)$ follows from Corollary \ref{cor3}, $1-e^{-x}<x$ for $x>0$, and $F_s(y)\leq1$ for $y\geq0$.}

Now, Theorems \ref{thm3} and \ref{thm4} are proved by substituting $P_{\rm{tx}}$ in Lemmas \ref{lem7} and \ref{lem8} into \eqref{eq:118}, respectively.



\section{Conclusion}\label{sec7}
In this paper, we showed that the node mobility greatly improves the throughput scaling of the covert communication over a wireless adhoc network. In particular, the aggregate throughput scaling was shown to be linear in $n$ when the number of channels that each warden uses to judge the presence of communication is not too large compared to~$n$. For achievability, we proposed a mobility-assisted scheme  where the communication from a source to its destination consists of two-hop small-range transmission. This scheme was shown to be optimal for $0<s<1$ under the assumption that each node distant from every warden to a certain extent uses the same power at each channel use. 

We note that our model assumes some impractical situations for the simplicity of analysis. First, it is assumed that the nodes are uniformly and independently distributed in each time $t$. In practice, each node has a correlated trajectory, like random walk model \cite{Gamal:2006-1}. In the case without covertness constraint, it is known that several constraints of node trajectory do not severely affect the throughput scaling. For example, the aggregate throughput still scales linearly in $n$ even if the trajectory of each node is restricted by a random line segment \cite{Diggavi:2005}. Similarly, we conjecture that the throughput scaling will not be severely affected by limited correlation of trajectory even in the presence of the covertness constraint. Second, the delay toleration from source to destination is assumed to be sufficiently large. The trade-off  between the delay toleration and the aggregate throughput was studied in the absence of covertness constraint \cite{Gamal:2006-1, Gamal:2006-2}. It would an interesting further work to study the effect of delay toleration constraint on the covert communication over the wireless adhoc network.  

Finally, we think that proving a nontrivial upper bound without the assumption of equal transmit power would be a good further work. It seems to be challenging  since the distances between the senders and the wardens, which are related to the upper bound on the transmit power from the covertness constraint, and the distances between the senders and the receivers, which affect the transmission rate, independently vary over time. 

\bibliographystyle{IEEEtran}
\bibliography{refs_all}

\end{document}